\documentclass[11pt]{article}

\usepackage[margin=1in]{geometry}
\usepackage{amsmath,amssymb,amsthm,mathtools}
\usepackage[T1]{fontenc}
\usepackage{microtype}
\usepackage{braket}
\usepackage{xcolor}
\usepackage{mathrsfs}
\usepackage[numbers, sort&compress]{natbib}
\usepackage{makecell}
\usepackage{booktabs}
\usepackage[blocks]{authblk}
\usepackage[
    colorlinks,
    allcolors=blue,
    pdfauthor={Steven T Flammia, Savar D Sinha, Yu Tong},
    pdftitle={Gauge freedom and efficient algorithms for Lindbladian learning}
]{hyperref}

\makeatletter
\renewcommand\AB@authnote[1]{\raisebox{0.6ex}{\normalfont\scriptsize#1}}
\renewcommand\AB@affilnote[1]{\raisebox{0.6ex}{\normalfont\scriptsize#1}}
\makeatother

\newcommand{\supp}{\operatorname{supp}}
\renewcommand{\Re}{\operatorname{Re}}
\renewcommand{\Im}{\operatorname{Im}}

\newcommand{\Tr}{\mathrm{Tr}}

\newcommand{\sgn}{\operatorname{sgn}}
\newcommand{\Or}{\mathcal{O}}

\newtheorem{theorem}{Theorem}
\newtheorem{corollary}[theorem]{Corollary}
\newtheorem{lemma}[theorem]{Lemma}
\newtheorem{proposition}[theorem]{Proposition}
\newtheorem{definition}[theorem]{Definition}
\newtheorem{remark}[theorem]{Remark}
\newtheorem{fact}[theorem]{Fact}

\title{Gauge freedom and efficient algorithms for Lindbladian learning}
\author[1,2]{Steven T. Flammia}
\author[3,4]{Savar D. Sinha}
\author[5,6,7]{Yu Tong}

\affil[1]{Department of Computer Science, Virginia Tech}
\affil[2]{Institute for Advanced Computing, Virginia Tech}
\affil[3]{Department of Computing and Mathematical Sciences, California Institute of Technology}
\affil[4]{Division of Physics, Mathematics and Astronomy, California Institute of Technology}
\affil[5]{Department of Mathematics, Duke University}
\affil[6]{Department of Electrical and Computer Engineering, Duke University}
\affil[7]{Duke Quantum Center, Duke University}
\date{\today}

\begin{document}
\maketitle

\begin{abstract}
    We study the problem of learning a local Lindbladian in the presence of state-preparation-and-measurement (SPAM) noise. Although recent work has developed scalable learning algorithms under idealized access assumptions, SPAM can make distinct Lindbladians experimentally indistinguishable, thus making part of the Lindbladian fundamentally unlearnable. We give a sharp characterization of this obstruction for bounded-degree local Lindbladians. We identify a family of locality-preserving gauge transformations that commutes with trusted single-qubit control, and use it to classify Lindbladian components, i.e., {Hamiltonian coefficients and the real and imaginary parts of dissipative coefficients, into universally gauge-invariant and generically gauge-dependent ones. The generic dependence holds among physical Lindbladians with each prescribed support pattern}. For the {generically} gauge dependent components, we construct Lindbladians to show that they can have $\Omega(1)$ uncertainty independent of system size, even after imposing complete positivity.

    We complement this characterization with SPAM-robust algorithms for learning every {universally gauge-invariant} component. Our algorithms use only trusted single-qubit operations, require neither ancillas nor entangling control, and use $\widetilde{\mathcal O}\left(\epsilon^{-2}\log(N/\delta)\right)$ experiments and total evolution time for constant locality and degree. Our guarantees allow global SPAM to become arbitrarily far from ideal as $N$ grows, requiring only nonvanishing local visibility. {We also establish sufficient conditions for boundary-induced gauge invariance of additional Hamiltonian coefficients under physical positivity constraints, though learning these additional parameters in general remains an open problem.} More broadly, our gauge-aware framework and SPAM-cancellation techniques offer a practical toolkit for the scalable characterization of open quantum systems.
\end{abstract}

\tableofcontents

\section{Introduction}

Learning a quantum system from its dynamics is a basic inverse problem in
quantum science.  For an open system, the object to be learned is its
Lindbladian: the generator that contains both the coherent interactions within the quantum system and the dissipative processes induced by the environment.
An accurate description of this generator is useful for diagnosing and calibrating quantum processors, validating analog and digital quantum simulators, constructing noise models for error mitigation, and understanding or engineering dissipative many-body dynamics~\cite{samach2022lindblad,carrasco2021,vandenberg2023pec,
barreiro2011open,verstraete2009quantum}.  
These tasks are becoming increasingly important as quantum devices grow to a size that makes full process tomography prohibitively expensive.

The problem has a long history.  
Early approaches reconstructed Lindblad operators from tomographic snapshots at several evolution times~\cite{boulant2003}.  
Later work exploited locality and many-body structure,
for example by inferring local generators from steady states and local
measurements \cite{bairey2020learning}, while experiments began fitting Lindblad models directly to processor data
\cite{samach2022lindblad,onorati2023fitting}.  
More recently, there has been rapid progress toward provably efficient methods for learning local
Lindbladians from dynamical data
\cite{olsacher2025liouvillian,arad2026near,mobus2026robust,
lewis2026structure,song2026constant,chenYu2026arbitrary,romanov2026qec}.

Concretely, we study an $N$-qubit Lindbladian expanded in the Pauli basis as
\begin{equation}
\label{eq:intro_lindbladian}
\begin{aligned}
    \mathcal L[\rho]
    ={}-i\sum_{c\in\mathcal I_h}h_c[P_c,\rho]
    +\sum_{(a,b)\in\mathcal I_d}\alpha_{ab}\mathcal D_{ab}[\rho],\quad \mathcal D_{ab}[\rho]
    =P_a\rho P_b-\frac{1}{2}\{P_bP_a,\rho\}.
\end{aligned}
\end{equation}
Here the first sum for $\mathcal{L}[\rho]$ is the coherent, or Hamiltonian, part, while the second is
the dissipative part, and the Kossakowski matrix
$A=(\alpha_{ab})$ is positive semidefinite.  The learning problem is to
recover the coefficients $h_c$ and $\alpha_{ab}$ from experiments involving
the evolution $e^{t\mathcal L}$. {We assume knowledge of a bounded-degree local support pattern (see Definition~\ref{def:support_pattern}), specifying which patches of qubits are allowed to interact. The candidate index sets \(\mathcal I_h\) and \(\mathcal I_d\) include, respectively, every nonidentity Pauli operator supported within an allowed patch and every ordered pair of nonidentity Pauli operators whose joint support is contained in an allowed patch. For fixed locality and degree, these candidate lists have total size \(\mathcal O(N)\); their coefficients are unknown and may vanish.} Moreover, we assume $|h_c|,|\alpha_{ab}|\leq 1$.

There is, however, a basic gap between an idealized learning protocol and an
experiment.  The states used to probe a device are prepared imperfectly, and
the final measurements are themselves noisy.  
These state-preparation-and-measurement (SPAM) errors can potentially affect the learnability of the Lindbladian: different combinations of SPAM noise channels and Lindbladians may result in exactly the same experimental outcome, making them indistinguishable from each other.

This obstruction also arises when one tries to characterize the noise process involved in quantum gates, and has motivated the development of gate-set tomography (GST).  By estimating
states, measurements, and gates self-consistently, GST avoids assigning all
observed error to the gate under study, but it recovers the gate set only up
to a similarity, or \emph{gauge}, transformation
\cite{merkel2013self,BlumeKohout2013GST,greenbaum2015introduction,nielsen2020gate}. 
Related gauge freedoms also determine which features of Pauli noise can be learned in the
presence of noisy SPAM~\cite{huang2022foundations}. 
For Clifford circuits, the learnable and unlearnable
directions can be identified with the cycle and cut spaces of a pattern
transfer graph \cite{chen2023learnability, chen2026efficient}.

In the continuous-time setting that we focus on, an experiment consisting of state preparation, evolution under $\mathcal{L}$, and measurement has output probabilities of the form
\begin{equation}
    p(t)=\Tr\!\left[M e^{t\mathcal L}(\rho_0)\right].
\end{equation}
For an invertible map $\mathcal S$, the simultaneous transformation
\begin{equation}
    \rho_0\mapsto\mathcal S(\rho_0),\qquad
    M\mapsto(\mathcal S^{-1})^\dagger(M),\qquad
    \mathcal L\mapsto\mathcal S\circ\mathcal L\circ\mathcal S^{-1}
\end{equation}
leaves every such probability unchanged whenever the transformed objects
remain admissible (i.e., being valid states, Lindbladians, and positive semidefinite operators respectively).  Consequently, no amount of data can distinguish two
physical models related in this way.  This raises the question that motivates
our work:

\begin{quote}
\emph{Which components of a local Lindbladian are identifiable in the presence of unknown SPAM noise, and can the identifiable components be learned efficiently using only trusted single-qubit control?}
\end{quote}

We focus on trusted single-qubit gates because they are typically easier to calibrate and deploy than entangling gates, whose characterization may itself be comparable in difficulty to learning the unknown many-body dynamics \cite{arute2019quantum}. 
For example, trapped-ion processors can implement single-qubit operations with extremely high fidelity~\cite{harty2014high,smith2025single}, whereas entangling gates require characterizing and controlling collective motional modes~\cite{choi2014optimal}.
Restricting attention to single-qubit operations avoids replacing the characterization task with an equally difficult calibration problem, while keeping the protocols applicable when reliable entangling gates or ancillary qubits are unavailable. 

\subsection{Our results}

We answer this question for bounded-degree local Lindbladians on $N$ qubits.
For SPAM, we will only prepare the all-zero initial state and measure single qubits in the computational ($Z$) basis. 
Our precise assumption on the SPAM noise can be found in Definition~\ref{defn:spam_noise_model}. We note that we do not require the SPAM noise channels to take any specific form (e.g., depolarizing noise as in Ref.~\cite{ivashkov2026ansatz}). 
We do not require
the global prepared state or the global measurement channel to be close to
ideal; in fact, they may remain arbitrarily far from the ideal operations as the system size grows. 
Instead, we only require the exact state and measurement to be ``visible'' in the actual noisy state and measurement we implement for constant-sized subsets of qubits.

\paragraph{Summary of our results.}
Subject to the above weak SPAM assumption, {the main results of this work can be summarized as:
For any fixed specification of which qubits can interact with which, we classify every component of a local Lindbladian into two categories:  those that are universally gauge-invariant and those that are generically gauge-dependent.
For the first category, we design 
efficient learning protocol to learn all of them using trusted single-qubit control despite unknown SPAM noise. For the second category, such SPAM noise can prevent arbitrarily precise learning, thus creating a fundamental error floor, for almost every physical Lindbladian.}

\paragraph{Gauge freedom and identifiability.}
We first construct an $N$-parameter family of gauge transformations from
single-qubit depolarizing maps.  These transformations commute with every
trusted single-qubit unitary, preserve locality, and can be absorbed entirely
into the unknown preparation and measurement operations.  We also characterize
the more general family of gauge transformations that commute with arbitrary
single-qubit control.

In the Pauli basis, the local depolarizing gauge gives a sharp classification
of the Lindbladian components.  
{Here a \emph{component} means a Hamiltonian coefficient or the real or imaginary part of a dissipative coefficient.}
The real and imaginary parts are separated because they may not be simultaneously gauge dependent or invariant.
We will need to use several slightly different concepts of gauge invariance and dependence, as rigorously defined in Definition~\ref{def:generic_gauge_dependence}. We provide a high-level explanation here:
{
We use the local depolarizing gauge transformations
\(\mathcal S_\theta\) defined in Eq.~\eqref{eq:gauge_on_pauli}.
The \emph{ambient generator space} \(\mathcal V\) includes both
physical Lindbladians and generators that need not produce physical
dynamics.
}

{
\begin{enumerate}
    \item \emph{Universal gauge invariance.}
    A component is universally gauge invariant if its value is unchanged
    by every gauge transformation, for every generator in the ambient
    space.
    \item \emph{Gauge dependence at a given generator.}
    A component is gauge dependent at a given generator in the ambient
    sense if some gauge transformation changes its value, without
    requiring the transformed generator to remain physical.
    \item \emph{Generic gauge dependence in the ambient sense.}
    A component is generically gauge dependent in the ambient sense
    if such changes are possible for almost every ambient generator.
    We use this notion as an intermediate step in the proof.
    \item \emph{Generic gauge dependence among physical Lindbladians.}
    Fix which groups of qubits a Hamiltonian or dissipative term may
    act on; we call this a \emph{support pattern}.
    A component is generically gauge dependent among physical
    Lindbladians if, for almost every physical Lindbladian with this
    support pattern, arbitrarily small gauge transformations can
    change its value while preserving both physicality and the
    support pattern.
\end{enumerate}
}

{
In both genericity statements, the relevant set is open, dense, and
of full measure, as made precise in
Definition~\ref{def:generic_gauge_dependence}.
Unless otherwise specified, \emph{generic gauge dependence} refers
to the physical notion in item~4.
}

{
Theorem~\ref{thm:gauge_invariant_coefficients} establishes the ambient
classification, and
Corollary~\ref{cor:generic_physical_gauge_dependence} strengthens it
to the physical setting for every component allowed by the prescribed
support pattern. Together, these results give the following
classification, in which we show that every Lindbladian component is either universally gauge invariant (item~1) or generically gauge dependent among physical Lindbladians (item~4).
}

\begin{theorem}[Gauge freedom and identifiability, informal version of Theorems~\ref{thm:gauge_invariant_coefficients} and \ref{thm:uniform_gauge_width} {and Corollary~\ref{cor:generic_physical_gauge_dependence}}]
Consider a bounded-degree local Lindbladian in the presence of unknown SPAM
noise and trusted single-qubit control.  There is an $N$-parameter family of
local depolarizing gauge transformations that leaves all experimental data
unchanged.  {The universally gauge-invariant components are exactly}
\begin{enumerate}
    \item the diagonal coefficients $\alpha_{aa}$;
    \item (Type~I) $\Re\alpha_{ab}$ when $P_a$ and $P_b$ differ on exactly one qubit
    and both are nonidentity on that qubit; and
    \item (Type~II) $\Im\alpha_{ab}$ when $P_a$ and $P_b$ differ on exactly one qubit
    and exactly one is the identity on that qubit, and $h_c$ when $P_c$ is single-qubit.
\end{enumerate}
Every other component allowed by locality is {generically gauge dependent among physical Lindbladians with the prescribed support pattern}.  Moreover, there exist bounded-degree local Lindbladians for which, for each component outside the diagonal, Type I, and Type II classes, one can find a gauge transformation that changes that component by an $\Omega(1)$ amount while preserving the validity of the transformed Lindbladian.
\end{theorem}

{By Corollary~\ref{cor:generic_physical_gauge_dependence}, all generically gauge-dependent components allowed by locality can vary simultaneously on an open dense, full-measure subset of physical Lindbladians with a prescribed support pattern, including under the coefficient bounds. However, at particular boundary instances, the requirement that the generator produce a completely positive and trace-preserving (CPTP) semigroup can force additional coefficients to remain unchanged among physically admissible gauge-equivalent Lindbladians. We call this \emph{boundary-induced gauge invariance}, and distinguish it from universal gauge invariance.}
A notable example is when the Lindbladian consists only of a Hamiltonian part. 
In this case all the Hamiltonian coefficients are SPAM-robustly learnable~\cite{HuangTongFangSu2023learning,yu2023robust}.

For the generically gauge dependent components, in fact a stronger statement can be made, as stated at the end of the theorem: one can find a Lindbladian for which all these components have $\Omega(1)$ uncertainty, which is independent of the system size even under the CPTP constraints. 
Moreover, this is true even if we specify the support of the Lindbladian terms, e.g., requiring it to be local on a graph. For a precise statement see Theorem~\ref{thm:uniform_gauge_width}.

\paragraph{{Efficient learning of the universal invariants.}}
We will show that {all universally} gauge-invariant components (diagonal dissipative coefficients and Type~I and II components) can be estimated with nearly optimal statistical scaling {(within polylogarithmic factor of the sample and total evolution time lower bounds established in \cite{arad2026near})}.
{We measure efficiency using three metrics: the \emph{number of experiments}, the \emph{total evolution time}, and the \emph{total number of single-qubit gates}.} By total evolution time, we mean the cumulative time spent evolving under the unknown Lindbladian across all experiments. {Since each experiment uses at most \(\Or(1)\) evolution time, our experiment-count bounds also upper-bound the total evolution time up to a constant factor.}
{Gate counts include all single-qubit control operations across all qubits and experiments.}

For the diagonal coefficients, Pauli twirling converts the relevant local signals into exponential decay curves. Unknown state-preparation and measurement factors enter only as multiplicative prefactors and therefore cancel when the decay rates are extracted. Locality then allows the coefficients to be recovered efficiently from expectation values of local observables. Type I coefficients are learned similarly, after supplementing the Pauli twirl with suitable non-Pauli single-qubit rotations. 

{
Type II coefficients require a different approach because they cannot, in general, be converted into decay rates. Instead, we perform a \emph{partial twirl} to isolate the dynamics in the Heisenberg picture to a two-dimensional invariant subspace. For intuition, consider a single qubit Lindbladian. The  partial twirl through randomly applying \(\{I,\sigma^z\}\) results in an effective Lindbladian
\begin{equation}
    \label{eq:single_qubit_partial_twirl_example}
\begin{aligned}
\mathcal L_{\mathrm{tw}}[\rho]
=-ih_z[\sigma^z,\rho]
+\sum_{\mu\in\{x,y,z\}}\alpha_{\mu\mu}(\sigma^\mu\rho\sigma^\mu-\rho)
+\alpha_{xy}\left(\sigma^x\rho\sigma^y+\frac{i}{2}\{\sigma^z,\rho\}\right)
+\alpha_{xy}^{*}\left(\sigma^y\rho\sigma^x-\frac{i}{2}\{\sigma^z,\rho\}\right).
\end{aligned}
\end{equation}
The span of \(\sigma^x,\sigma^y\) is invariant in the Heisenberg picture: writing \(\sigma^\mu(t)=e^{t\mathcal L_{\mathrm{tw}}^\dagger}[\sigma^\mu]\), $\mu=x,y,z$, we obtain
$$
\frac{d}{dt}
\begin{pmatrix}\sigma^x(t)\\\sigma^y(t)\end{pmatrix}
=A\begin{pmatrix}\sigma^x(t)\\\sigma^y(t)\end{pmatrix}
:=
2\begin{pmatrix}
-(\alpha_{yy}+\alpha_{zz}) & \operatorname{Re}\alpha_{xy}-h_z\\
\operatorname{Re}\alpha_{xy}+h_z & -(\alpha_{xx}+\alpha_{zz})
\end{pmatrix}
\begin{pmatrix}\sigma^x(t)\\\sigma^y(t)\end{pmatrix}.
$$
Thus the single-qubit Type II parameter \(h_z\) appears in the antisymmetric coupling between \(\sigma^x\) and \(\sigma^y\), while the dissipative contribution \(\operatorname{Re}\alpha_{xy}\) enters symmetrically.
By applying random Pauli operators and correcting for the resulting sign in the state preparation and measurement stages, we can average measurement results to reconstruct the matrix $mse^{tA}$ where $m,s$ are scalar factors from measurement and state preparation noises respectively. Thus the effect of SPAM noise is reduced to a single scalar factor $ms$, which can be estimated at $t=0$.
Because of this we can extract $h_z$ from derivatives at \(t=0\) using robust polynomial interpolation in Lemma~\ref{lem:chebyshev_derivative_estimation}.
This general methodology works beyond the single-qubit setting, as shown in \eqref{eq:typeii_effective_action}, and can be parallelized across the entire system, leading to the learning time stated in the next theorem:
}

\begin{theorem}[Efficient SPAM-robust learning, informal version of Theorems~\ref{thm:estimate_diagonal}, \ref{thm:estimate_typei}, and \ref{thm:estimate_typeii}]
We consider a bounded-degree local Lindbladian $\mathcal{L}$ of fixed locality and degree.  Under our SPAM model in Definition~\ref{defn:spam_noise_model}, there are protocols using
only trusted single-qubit control that estimate every {universally} gauge-invariant
Lindbladian component 
to additive error $\epsilon$ with failure probability at most $\delta$.  All
diagonal coefficients and all Type~I components can be
learned using
$\Or\!\left(\epsilon^{-2}\log(N/\delta)\right)$
experiments and total evolution time.  All Type~II components can be learned
using
$
    \Or\!\left(
        \epsilon^{-2}\log^4(2/\epsilon)\log(N/\delta)
    \right)
$
experiments and total evolution time of the same order.  
{The total single-qubit gate count is
\(\Or(N\epsilon^{-3}\log(N/\delta))\) for the diagonal and Type~I
protocols, and
\(\Or(N\epsilon^{-3}\log^6(2/\epsilon)\log(N/\delta))\) for the
Type~II protocol.}
\end{theorem}

\paragraph{{Boundary-induced gauge invariance of Hamiltonian coefficients.}}
{A generically gauge-dependent component may nevertheless be gauge invariant at a particular boundary Lindbladian: 
any gauge transformation either transforms the boundary Lindbladian to something unphysical, or leaves the component value unchanged.
We use positivity of the Kossakowski matrix to give efficiently checkable sufficient conditions for this boundary-induced invariance of Hamiltonian coefficients under gauge transformations commuting with single-qubit control. These additional invariants depend on the instance and are distinct from the universally gauge-invariant components learned above. Learning these additional parameters in general remains an open problem.}

\begin{theorem}[{Boundary-induced Hamiltonian gauge invariance}, informal version of Theorems~\ref{thm:hamiltonian_product_condition}
and~\ref{thm:hamiltonian_odd_weight_condition}]
Let $\mathcal L$ and $\mathcal L'$ be physical Lindbladians related by a gauge
transformation that commutes with all single-qubit control.  A Hamiltonian
coefficient $h_c$ is gauge invariant if no two Pauli operators with nonzero
diagonal Kossakowski coefficients have projective product $P_c$, provided the
gauge factors are positive.  When $P_c$ has odd weight, it is enough to check
only pairs obtained by restricting $P_c$ to complementary subsets of its
support, and no positivity assumption on the gauge factors is needed. 
\end{theorem}

As an example application of this theorem, let us consider a Lindbladian
\[
\mathcal{L}[\rho] = -i[H,\rho] + \gamma\sum_{i}(\sigma^x_i\rho \sigma^x_i-\rho), \quad
H=\sum_{i} J_{i,i+1}\sigma^z_i \sigma^z_{i+1},
\]
then the coefficients $J_{i,i+1}$ {exhibit boundary-induced gauge invariance for positive gauge factors} because multiplying two Pauli operators appearing in the dissipative part, i.e., $\sigma_i^x$ in this case, cannot give us $\sigma^z_i \sigma^z_{i+1}$. This is reminiscent of the Hamiltonian-not-in-Lindbladian-span condition for the attainability of the Heisenberg limit \cite{Zhou2017AchievingTH}.

\subsection{Related work on SPAM noise and control assumptions}
\label{sec:related_spam}

\begin{table}[t]
\centering
\small
\setlength{\tabcolsep}{3pt}
\renewcommand{\arraystretch}{1.18}

\begin{tabular}{
p{0.22\linewidth}
p{0.18\linewidth}
p{0.15\linewidth}
p{0.20\linewidth}
p{0.18\linewidth}}
\toprule
\textbf{Work}
&
\textbf{\makecell[l]{SPAM noise}}
&
\textbf{Scalable}
&
\textbf{\makecell[l]{Classification \\ of learnability}}
&
\textbf{Trusted control}
\\
\midrule

Stilck Fran\c{c}a \emph{et al.}~\cite{StilckFrança2024}
&
Known
&
Yes
&
None
&
Single-qubit
\\

Liu \emph{et al.}~\cite{liu2025robust}
&
Unknown
&
No
&
Partial
&
None
\\

Ivashkov \emph{et al.}~\cite{ivashkov2026ansatz}
&
Known
&
Conditional
&
None
&
Single-qubit
\\

Sinha~\cite{sinha2026efficient}
&
Unknown
&
Yes
&
Full
&
Multi-qubit
\\

Möbus \emph{et al.}~\cite{mobus2026robust}
&
Ideal
&
Yes
&
None
&
Single-qubit
\\

Lewis \emph{et al.}~\cite{lewis2026structure}
&
Ideal
&
Yes
&
None
&
Single-qubit
\\

Zhou and Gong~\cite{zhouGong2026few}
&
Known
&
Yes
&
None
&
Single-qubit
\\

\emph{This work}
&
Unknown
&
Yes
&
Full
&
Single-qubit\\
\bottomrule
\end{tabular}
\caption{Comparison of SPAM noise assumptions, scalability, identifiability results, and trusted-control requirements of selected Lindbladian learning works. 
For SPAM noise assumptions, ``ideal'' means the work assumes ideal SPAM without noise; ``known'' means the work assumes known SPAM channel, which can also be restricted to a specific form; ``unknown'' means no prior assumption on the form of the SPAM noise channels are assumed. 
Scalability here means whether learning a Lindbladian on $N$ qubits, under sparse or bounded-degree local assumption, can be accomplished in time $\mathrm{poly}(N)$. The scalability guarantee of Ref.~\cite{ivashkov2026ansatz} is conditional because it depends on a conditioning factor without a priori bound.
For classification of learnability, ``full'' denotes a classification of every coefficient within the stated model; ``partial'' denotes a characterization of only selected components or a global gauge freedom; ``none'' denotes the absence of an unknown-SPAM identifiability analysis. {For this work, ``full'' refers to the classification into universally gauge-invariant and generically gauge-dependent components; learning additional boundary-induced invariants in general remains open.}}
\label{tab:spam-control-comparison}
\end{table}

Prior approaches to Lindbladian learning differ substantially in how they
treat state-preparation-and-measurement (SPAM) noise.  
In Table~\ref{tab:spam-control-comparison} we compare a selected set of Lindbladian learning works in terms of their assumption on the SPAM noise, whether they can efficiently scale to large systems, whether they provide a full classification of the learnability of Lindbladian components, and what forms of trusted control they use. To the best of our knowledge, our work is the only one that addresses SPAM noise of unknown form, is scalable for large systems (having only $\Or(\log(N))$ overhead for $N$ qubits), provides {a complete classification into universally gauge-invariant and generically gauge-dependent Lindbladian components}, and only uses trusted single-qubit control.

Below we will group the relevant works in terms how they address the SPAM noise, and provide a more detailed discussion.

\paragraph{Ideal SPAM.} All works in this category derive their main learning guarantees assuming either ideal state preparation or ideal measurement \cite{arad2026near,mobus2026robust,lewis2026structure,song2026constant,chenYu2026arbitrary,romanov2026qec,bairey2020learning,olsacher2025liouvillian,onorati2023fitting}. Ref.~\cite{olsacher2025liouvillian} does not require exact preparation of the intended initial states, but does not analyze unknown measurement noise. Ref.~\cite{onorati2023fitting} explicitly discusses
state-preparation and measurement errors, but does not provide a guarantee
relating such errors to the accuracy of the recovered Lindbladian.

\paragraph{Known SPAM noise channels.}
Several works account for SPAM errors by treating the corresponding noise channels as exactly known.  
Ref.~\cite{zhouGong2026few} assumes known Pauli-diagonal channels for both state preparation and measurement, while Ref.~\cite{ivashkov2026ansatz} specializes to known independent single-qubit depolarizing channels.  
One of the regimes analyzed in Ref.~\cite{StilckFrança2024} likewise assumes known SPAM channels that can be efficiently incorporated into the learning procedure. 
In practice, SPAM channels are generally not known a priori, because self-consistent characterization of state preparations and measurements determines them only up to gauge transformations unless additional operations are assumed to be trusted \cite{nielsen2020gate}.

\paragraph{Unknown SPAM without systematic identifiability analysis.} 
Unknown SPAM channels are often discussed in the literature without a detailed analysis of their effects.  
Ref.~\cite{samach2022lindblad} fits an initial SPAM noise model using maximum likelihood. 
It does not distinguish what degrees of freedom are learnable, and the runtime is exponential in the number of qubits. 
Ref.~\cite{StilckFrança2024} analyzes a regime with unknown SPAM noise that creates an error floor on achievable precision. 
Ref.~\cite{seif2021compressed} analyzes the effect of small unknown SPAM errors, showing that first-order contributions can be absorbed into the amplitude and offset of the decay signal.
However, these two works do not show which part of the Lindbladian is limited by SPAM in the achievable precision and which part can be learned to arbitrary precision. Ref.~\cite{dobrynin2025compressed} applies Lindbladian tomography to experimental data containing SPAM and acknowledges that SPAM limits the achievable accuracy. 
It does not isolate or quantify the resulting SPAM-induced bias.

\paragraph{Unknown SPAM with systematic identifiability analysis.}
Several works explicitly address learning in the presence of unknown SPAM noise, although under more restrictive dynamical models or stronger control assumptions. 
Ref.~\cite{liu2025robust} combine gate-set tomography with Lindbladian fitting to jointly estimate Lindbladian models for a collection of noisy gates and the associated state preparations and measurements, subject to the similarity-transformation gauge freedom.
It does not determine which individual Hamiltonian and dissipative coefficients remain identifiable, and its reliance on full process tomography makes the runtime exponential in the number of qubits. 
Refs.~\cite{vandenberg2023pec,vandenbergWocjan2024techniques,chen2023learnability,chen2026efficient} instead consider Pauli or Pauli-Lindblad noise associated with Clifford circuits. 
Among these, Ref.~\cite{chen2023learnability} gives a complete characterization of the learnable and unlearnable combinations of Pauli-noise parameters. 
Ref.~\cite{chen2026efficient} extends this characterization to structured Pauli-noise models and provides efficient protocols for local and quasi-local noise. 
Its experimental framework is closely related to averaged circuit eigenvalue sampling (ACES)~\cite{flammia2021averaged}, while explicitly incorporating unknown Pauli SPAM channels into the learnability analysis. 
These methods are restricted to Pauli-diagonal noise and do not cover general local Lindbladians containing coherent Hamiltonian terms and off-diagonal dissipative couplings. 
Ref.~\cite{sinha2026efficient} gives a complete classification, at the level of individual coefficients, of the learnable and unlearnable components of a sparse Lindbladian under access to arbitrary trusted multiqubit unitaries. 
By contrast, the present work requires only trusted single-qubit gates and determines which coefficients of a general local Lindbladian remain identifiable under this weaker control assumption.

\paragraph{Trusted control.}
SPAM assumptions should be distinguished from control requirements.  
The product-Pauli protocols in
Refs.~\cite{arad2026near,mobus2026robust,lewis2026structure,song2026constant,ivashkov2026ansatz,zhouGong2026few} require no trusted entangling gates.  
By contrast, Ref.~\cite{chenYu2026arbitrary,sinha2026efficient} use trusted multiqubit unitaries, and Ref.~\cite{romanov2026qec} requires ancillary qubits and interleaved stabilizer error correction.  
Our protocol only assumes
trusted single-qubit gates, and requires neither ancillary
qubits nor trusted multi-qubit control, while allowing unknown SPAM.

\section{Preliminaries}
\label{sec:preliminaries}

We denote the Pauli matrices by $\sigma^x,\sigma^y,\sigma^z$.
We denote  the $N$-qubit Pauli group by $\mathcal{P}_N$
\begin{equation}
    \mathcal{P}_N = \{i^{r}P_1\otimes P_2\otimes\cdots\otimes P_N: P_j\in \{I,\sigma^x,\sigma^y,\sigma^z\}, r=0,1,2,3\}.
\end{equation}
We denote by $\bar{\mathcal{P}}_N$ the projective Pauli group
\begin{equation}
    \bar{\mathcal{P}}_N = \mathcal{P}_N/\braket{i}.
\end{equation}
Because $\bar{\mathcal{P}}_N$ is isomorphic to the vector space $\mathbb{F}_2^{2N}$, the symplectic inner product in $\mathbb{F}_2^{2N}$ induces a symplectic inner product for $\bar{\mathcal{P}}_N$, which, for $\bar{P},\bar{Q}\in \bar{\mathcal{P}}_N$, can be defined as
\begin{equation}
\label{eq:Pauli_inner_product}
    \braket{\bar{P},\bar{Q}} = \begin{cases}
        0,\quad \text{if } PQ=QP,\\
        1,\quad \text{if } PQ\neq QP.\\
    \end{cases}
\end{equation}
Below, unless otherwise stated, we will not distinguish between $\bar{P}$ and its representative element $P$, which we will take to be the tensor product $\{I,\sigma^x,\sigma^y,\sigma^z\}$ without an additional $\braket{i}$ factor.

For a Pauli operator $P=P_1\otimes\cdots\otimes P_N$, we define its \emph{support vector} $s(P)\in\{0,1\}^N$ by
\begin{equation}
\label{eq:support_vector_defn}
    s_i(P)
    =
    \begin{cases}
        0, & P_i=I,\\
        1, & P_i\in\{\sigma^x,\sigma^y,\sigma^z\}.
    \end{cases}
\end{equation}

For any $S\subset [N]$, we define a superoperator embedding map $\iota_S$. 
For any superoperator $\mathcal{A}$ acting on $S$, we define
\begin{equation}
    \label{eq:embedding_map}
    \iota_S(\mathcal{A}) = \mathcal{A}\otimes \mathcal{I}_{S^c}.
\end{equation}

\section{Problem setup}
\label{sec:problem_setup}

We consider learning a Lindbladian on  $N$ qubits, labeled $a\in[N]$. 
The Lindbladian can be written as
\begin{equation}
\label{eq:lindbladian_defn}
    \mathcal{L}[\rho] = -i\sum_{c\in \mathcal{I}_h} h_c[P_c,\rho] + \sum_{(a,b)\in\mathcal{I}_d}\alpha_{ab} \mathcal{D}_{ab}[\rho],
\end{equation}
where
\begin{equation}
    \mathcal{D}_{ab}[\rho] = P_a\rho P_b-\frac{1}{2}\{{P_b P_a},\rho\}.
\end{equation}
{Here \(A=(\alpha_{ab})\) is the positive semidefinite Kossakowski matrix in the nonidentity Hermitian Pauli basis. We take \(h_c\in\mathbb R\) and fix the coefficient of the identity in the Hamiltonian to zero. The index sets \(\mathcal I_h,\mathcal I_d\) and their Pauli labels are known; only the coefficients are unknown, and they satisfy \(|h_c|,|\alpha_{ab}|\leq1\). Coefficients outside these lists are zero. The locality and degree bounds below apply every term with index in $\mathcal{I}_h$ and $\mathcal{I}_d$.}
We will call each $-ih_c[P_c,\cdot]$ and $\alpha_{ab}\mathcal{D}_{ab}$ a term in the Lindbladian, the former a \emph{coherent term} and the latter a \emph{dissipative term}. We will also define their support to be
\begin{equation}
    S_{c} = \mathrm{supp}(P_c),\quad S_{(a,b)} = \mathrm{supp}(P_a)\cup\mathrm{supp}(P_b).
\end{equation}
{We study the learnability of the coefficients \(h_c\) and \(\alpha_{ab}\). Since the real and imaginary parts of a dissipative coefficient may have different learnability properties, we refer to \(h_c\), \(\Re\alpha_{ab}\), and \(\Im\alpha_{ab}\) as Lindbladian \emph{components}.}

For Lindbladian terms, we will define an interaction graph
\begin{definition}[Interaction graph]
    The interaction graph $G_{\mathrm{int}}=(V_{\mathrm{int}},E_{\mathrm{int}})$ is the graph where $V_{\mathrm{int}}= \mathcal{I}_h\cup \mathcal{I}_d$, and for $\alpha,\beta\in V_{\mathrm{int}}$, $(\alpha,\beta)\in E_{\mathrm{int}}$ if and only if $S_\alpha \cap S_\beta \neq \varnothing$.
\end{definition}

We will assume that the Lindbladian is $(\mathsf{k},\mathsf{d})$-bounded degree local, with $\mathsf{k},\mathsf{d}=\Or(1)$, according to the following definition
\begin{definition}[Bounded-degree local]
    We say the Lindbladian $\mathcal{L}$ is $(\mathsf{k},\mathsf{d})$-bounded degree local if $\max_{\alpha\in V_{\mathrm{int}}}|S_\alpha|\leq \mathsf{k}$ and $G_{\mathrm{int}}$ has degree at most $\mathsf{d}$.
\end{definition}

\begin{definition}[Support pattern]
\label{def:support_pattern}
A \emph{support pattern} on \(N\) qubits is a finite family
$
    \mathsf{P}=\{C_1,\ldots,C_m\}
$
of nonempty subsets of \([N]\), called \emph{patches}. The
intersection graph of \(\mathsf{P}\) has one vertex for each patch,
with an edge between distinct patches \(C,C'\in\mathsf{P}\) exactly
when \(C\cap C'\neq\varnothing\). We say that \(\mathsf{P}\) is
\((\mathsf{k},\mathsf{d})\)-bounded degree local if
$
|C|\leq\mathsf{k}
$
for every $C\in \mathsf{P}$
and its intersection graph has degree at most \(\mathsf{d}\).
\end{definition}

\begin{definition}[Superoperators with given support pattern]
\label{defn:superops_with_support_pattern}
    For \(C\subseteq[N]\), let \(\mathcal{V}_C\) be the real vector space of Hermiticity-preserving, trace-annihilating superoperators on the qubits in \(C\). We define $\mathcal{V}_{\mathsf{P}}$ to be
    \begin{equation}
    \label{eq:local_generator_space}
        \mathcal{V}_{\mathsf{P}}
        =
        \sum_{C\in\mathsf{P}}\iota_C(\mathcal{V}_C).
    \end{equation}
    We denote the set of Lindbladians supported on \(\mathsf{P}\) by
\begin{equation*}
    \operatorname{Lind}(\mathsf{P})
    =
    \left\{
        \mathcal{L}\in\mathcal{V}_{\mathsf{P}}:
        \mathcal{L}\text{ is a Lindbladian}
    \right\}.
\end{equation*}
\end{definition}

If \(\mathsf{P}\) is
\((\mathsf{k},\mathsf{d})\)-bounded degree local, then every qubit belongs to at
most \(\mathsf{d}+1\) patches, since all patches containing that
qubit form a clique in the intersection graph. Consequently,
\begin{equation}
\label{eq:number_of_support_patches}
    |\mathsf{P}|
    \leq
    \sum_{C\in\mathsf{P}}|C|
    \leq
    (\mathsf{d}+1)N.
\end{equation}
In particular, \(|\mathsf{P}|=\Or(N)\) when
\(\mathsf{d}=\Or(1)\). The distinct nonempty term supports of a
\((\mathsf{k},\mathsf{d})\)-bounded-degree local Lindbladian form a
\((\mathsf{k},\mathsf{d})\)-bounded degree local support pattern.

We model the SPAM noise as follows
\begin{definition}[SPAM noise model]
\label{defn:spam_noise_model}
    Every experiment involves initializing the quantum system in state $\ket{0^N}$, and measuring all qubits in the computational basis at the end of the experiment. State preparation involves noise that corrupts the actual prepared state to be $\rho_0$, satisfying 
    \[
    \Tr[Z_C \rho_0]\geq r_{p},
    \]
    where $Z_C = \prod_{a\in C} \sigma^z_a$,
    for every $C\subset[N]$ such that $|C|\leq \mathsf{k}$. Measurement involves a noise channel $\mathcal{E}$ that corrupts the actual measured operator from $X$ to $\mathcal{E}^\dag[X]$, and it satisfies
    \[
    \overline{\Tr}[X \mathcal{E}^\dag[X]]\geq r_m,
    \]
    {for every Hermitian Pauli operator \(X\) with \(|\supp(X)|\leq\mathsf k\), where \(\overline{\Tr}\) is the normalized trace. We assume that \(\mathcal E\) is CPTP and that \(r_p,r_m>0\) are independent of the system size \(N\).}
\end{definition}
Our SPAM assumption is weaker than global closeness to ideal preparation and measurement. {We require nonvanishing preparation visibility for computational-basis Pauli parities and nonvanishing diagonal measurement visibility for Pauli operators of weight at most \(\mathsf k=\Or(1)\), with lower bounds independent of system size.} The global SPAM operations may be arbitrarily far from ideal as $N$ increases.

\section{Gauge transformations and gauge freedom}
\label{sec:gauge_transformations_freedom}

In this section, we introduce a class of gauge transformations that change Lindbladian components while leaving all experimental outcomes unchanged. We then classify the components as {universally gauge invariant or generically gauge dependent} and quantify the extent to which the gauge-dependent components can vary.

\subsection{Gauge transformations}
\label{sec:gauge_transformations}
For $\theta=(\theta_1,\ldots,\theta_N)\in\mathbb{R}^N$, we define a superoperator $\Gamma_\theta$ through
\begin{equation}
    \Gamma_\theta[X]
    =
    \frac{1}{4}\sum_{i=1}^N\theta_i(\sigma_i^x X \sigma_i^x+\sigma_i^y X \sigma_i^y+\sigma_i^z X \sigma_i^z-3X),
    \quad
    \mathcal{S}_\theta=e^{\Gamma_\theta}.
\end{equation}
By this definition, we have
\begin{equation}
\label{eq:gauge_on_pauli}
    \Gamma_\theta[P]
    =
    -\theta\cdot s(P)P,
    \quad
    \mathcal{S}_\theta[P]
    =
    e^{-\theta\cdot s(P)}P.
\end{equation}
We will define a class of gauge transformations through the similarity transformation
\begin{equation}
\label{eq:gauge_similarity}
    \mathcal{L}
    \longmapsto
    \mathcal{L}_\theta
    =
    \mathcal{S}_\theta
    \circ\mathcal{L}
    \circ\mathcal{S}_\theta^{-1}.
\end{equation}
This transformation can be absorbed into the unknown SPAM operations. Indeed, for an initial state $\rho_0$ and an effective measured observable $M$, define
\[
    \rho_{0,\theta}=\mathcal{S}_\theta[\rho_0],
    \quad
    M_\theta=(\mathcal{S}_\theta^{-1})^\dag[M].
\]
Then we have
\begin{equation}
    \Tr\left[
        M_\theta e^{t\mathcal{L}_\theta}[\rho_{0,\theta}]
    \right]
    = \Tr\left[
        M_\theta (\mathcal{S}_\theta^{-1}\circ \mathcal{S}_\theta \circ e^{t\mathcal{L}_\theta}\circ\mathcal{S}_\theta^{-1}\circ\mathcal{S}_\theta[\rho_{0,\theta}])
    \right]
    =
    \Tr\left[
        M e^{t\mathcal{L}}[\rho_0]
    \right].
\end{equation}
The same identity holds in the presence of controls that commute with
$\mathcal{S}_\theta$, including arbitrary products of single-qubit
unitary channels. 
In other words, for channels $\mathcal{U}_1,\mathcal{U}_2,\cdots, \mathcal{U}_r$ that are tensor products of single-qubit unitary channels,
\begin{equation}
\label{eq:indistinguishability_gauge_transform}
    \Tr\left[
        M_\theta e^{t_1\mathcal{L}_\theta}\mathcal{U}_1\cdots e^{t_r\mathcal{L}_\theta}\mathcal{U}_r[\rho_{0,\theta}]
    \right]
    =
    \Tr\left[
        M e^{t_1\mathcal{L}}\mathcal{U}_1\cdots e^{t_r\mathcal{L}}\mathcal{U}_r[\rho_{0}]
    \right].
\end{equation}
Thus, whenever the transformed preparation state and measurement observable
remain admissible under the SPAM noise model, the experimental data
cannot distinguish $\mathcal{L}$ from $\mathcal{L}_\theta$. 
{
\begin{remark}[Admissible SPAM for the gauge construction]
\label{rem:admissible_spam_gauge}
Let \(\mathcal D_\lambda[X]=\lambda X+(1-\lambda)\Tr[X]I/2\)
be the single-qubit depolarizing channel for \(0\leq\lambda\leq1\).
For any Lindbladian \(\mathcal L\), choose preparation and
measurement-noise channels
\(\mathcal E_{\mathrm p}=\mathcal E=\mathcal D_{e^{-2}}^{\otimes N}\),
with \(\rho_0=\mathcal E_{\mathrm p}[\ket{0^N}\!\bra{0^N}]\).
For every \(\theta\in[-1,1]^N\), the transformed channels are
\begin{equation*}
\begin{aligned}
    \mathcal E_{\mathrm p,\theta}
    =\mathcal S_\theta\circ\mathcal E_{\mathrm p}
    =\bigotimes_{j=1}^N\mathcal D_{e^{-2-\theta_j}},\quad 
    \mathcal E_\theta
    =\mathcal E\circ\mathcal S_\theta^{-1}
    =\bigotimes_{j=1}^N\mathcal D_{e^{-2+\theta_j}}.
\end{aligned}
\end{equation*}
All multipliers lie in \([e^{-3},e^{-1}]\), so both channels are
CPTP. Their local preparation and measurement visibility factors
are products of at most \(\mathsf k\) such multipliers. Thus the
original and transformed SPAM models satisfy
Definition~\ref{defn:spam_noise_model} with
\(r_p=r_m=e^{-3\mathsf k}\), independently of \(N\) for fixed
locality \(\mathsf k\). Whenever \(\mathcal L_\theta\) is also a
Lindbladian, Eq.~\eqref{eq:indistinguishability_gauge_transform}
shows that the original and transformed models give identical
outcome distributions for every allowed controlled experiment.
\end{remark}
The indistinguishability identity~\eqref{eq:indistinguishability_gauge_transform}
implies that a quantity can be identified uniformly over the allowed
model class only if it is invariant under gauge transformations for
which both the generator and SPAM remain admissible.
For the gauge-dependence constructions below, since
Remark~\ref{rem:admissible_spam_gauge} already supplies admissible
SPAM channels that make the original and transformed models
experimentally indistinguishable for every \(\theta\in[-1,1]^N\), 
we may therefore focus on the Lindbladian and its physical constraints.}

While the family $\mathcal S_\theta$ suffices for a generic characterization of gauge freedom, as we will see in Section~\ref{sec:gauge_transformations_freedom}, it is not the most general gauge transformation
that commutes with all single-qubit unitary channels.  We recall the standard characterization of all such transformations \cite{chen2020robust,helsen2023shadow}.  For each
$S\subseteq[N]$, let $\Pi_S$ denote the linear projection onto the Pauli
operators with exact support $S$; equivalently,
\begin{equation}
    \Pi_S[P]
    =
    \begin{cases}
        P, & \supp(P)=S,\\
        0, & \supp(P)\neq S.
    \end{cases}
    \label{eq:exact_support_projection}
\end{equation}

\begin{proposition}[Commuting gauge transformations]
\label{prop:commuting_gauge_transformations}
A complex-linear superoperator $\mathcal G$ commutes with every single-qubit unitary channel acting on any one of the \(N\) qubits if and only if
\begin{equation}
    \mathcal G
    =
    \sum_{S\subseteq[N]}g_S\Pi_S,
    \qquad
    g_S\in\mathbb C.
    \label{eq:general_commuting_gauge}
\end{equation}
Equivalently, $\mathcal G[P]=g_{\supp(P)}P$ for every Pauli operator $P$.
If $\mathcal G$ is Hermiticity preserving, then every $g_S$ is real.  It is
invertible if and only if every $g_S$ is nonzero, while trace preservation and
unitality each require $g_\varnothing=1$.
\end{proposition}

This proposition follows directly from the known 
decomposition of the local-Clifford conjugation representation into mutually
inequivalent irreducible subspaces indexed by exact Pauli support \cite{helsen2023shadow}; the same
decomposition is used in the analysis of local-Clifford shadow
protocols~\cite{chen2020robust}.  Schur's lemma implies that a map commuting
with all local Clifford channels acts as a scalar on each exact-support
subspace.  Conversely, each product of single-qubit unitary channels preserves
exact Pauli support, so every map in Eq.~\eqref{eq:general_commuting_gauge}
commutes with all such channels.  Hence there are no additional complex-linear
commuting gauge transformations.

The gauge transformation $\mathcal S_\theta$ considered above is the
$N$-parameter subfamily for which
$
    g_S
    =
    \exp\left(-\sum_{j\in S}\theta_j\right),
$
and $g_\varnothing=1$. 
Equation~\eqref{eq:general_commuting_gauge} also includes a depolarizing map
acting jointly on any collection $A\subseteq[N]$ of qubits.  
The indistinguishability identity \eqref{eq:indistinguishability_gauge_transform}
above applies to every invertible $\mathcal G$ in
Eq.~\eqref{eq:general_commuting_gauge} whenever the transformed preparation
state and measurement observable remain admissible under the SPAM noise model.

As discussed above, the restricted family of gauge transformations \(\mathcal S_\theta\) suffices for the analysis in Section~\ref{sec:gauge_transformations_freedom}. Any component that is {generically} gauge-dependent under this restricted family is necessarily {generically} gauge-dependent under the full class of gauge transformations. {For all universally gauge-invariant components in this classification, Theorems~\ref{thm:estimate_diagonal}, \ref{thm:estimate_typei}, and~\ref{thm:estimate_typeii} give learning protocols, which also establish invariance under every commuting gauge transformation that preserves the admissibility of the generator and SPAM.} The more general gauge transformations beyond \(\mathcal S_\theta\) will be needed in Section~\ref{sec:hamiltonian_learnability}.

\subsection{Gauge-dependence of Lindbladian components}
\label{sec:gauge_dependence}
We next determine the action of this gauge transformation on the
Pauli-basis coefficients of $\mathcal{L}$. Write the Pauli-transfer
representation as
\begin{equation}
\label{eq:ptm_representation}
    \mathcal{L}[\rho]
    =
    \frac{1}{2^N}
    \sum_{P,Q\in\bar{\mathcal{P}}_N}
    T_{P,Q}P\Tr[Q\rho].
\end{equation}
Because the representatives of $\bar{\mathcal{P}}_N$ are Hermitian,
the coefficients $T_{P,Q}$ are real whenever $\mathcal{L}$ is
Hermiticity preserving. Equation~\eqref{eq:gauge_on_pauli} gives
\begin{equation}
\label{eq:ptm_gauge_transform}
    T_{P,Q}
    \longmapsto
    e^{\theta\cdot(s(Q)-s(P))}T_{P,Q}.
\end{equation}

Equivalently, write $\mathcal{L}$ in the Pauli-basis $\chi$
representation as
\begin{equation}
\label{eq:chi_representation}
    \mathcal{L}[\rho]
    =
    \sum_{U,V\in\bar{\mathcal{P}}_N}
    \chi_{U,V}U\rho V.
\end{equation}
For $A,B\in\bar{\mathcal{P}}_N$, let $A\star B$ denote the chosen
Hermitian representative of the product of their projective Pauli
labels, and define $\omega(A,B)\in\{\pm1,\pm i\}$ by
\begin{equation}
    AB=\omega(A,B)(A\star B).
\end{equation}
The Pauli-twirling identity
\begin{equation}
    \frac{1}{4^N}
    \sum_{R\in\bar{\mathcal{P}}_N}
    RXR
    =
    \frac{\Tr[X]}{2^N}I
\end{equation}
implies
\begin{equation}
    \frac{1}{2^N}P\Tr[Q\rho]
    =
    \frac{1}{4^N}
    \sum_{R\in\bar{\mathcal{P}}_N}
    RQ\rho RP.
\end{equation}
Making the change of variables
$
U=R\star Q
$
and
$
V=R\star P
$
therefore gives
\begin{equation}
\label{eq:chi_from_ptm}
    \chi_{U,V}
    =
    \frac{1}{4^N}
    \sum_{R\in\bar{\mathcal{P}}_N}
    \varphi_R(U,V)
    T_{R\star V,R\star U},
\end{equation}
where
\begin{equation}
\label{eq:chi_gauge_transform_phase}
    \varphi_R(U,V)
    =
    \omega(R,R\star U)\omega(R,R\star V).
\end{equation}
Combining \eqref{eq:ptm_gauge_transform} and
\eqref{eq:chi_from_ptm}, we obtain
\begin{equation}
\label{eq:chi_gauge_transform}
    \chi_{U,V}
    \longmapsto
    \frac{1}{4^N}
    \sum_{R\in\bar{\mathcal{P}}_N}
    \varphi_R(U,V)
    e^{\theta\cdot
    (s(R\star U)-s(R\star V))}
    T_{R\star V,R\star U}.
\end{equation}

For two Pauli operators $U$ and $V$, define
\begin{equation}
    D(U,V)=\{i\in[N]:U_i\neq V_i\}
\end{equation}
and
\begin{equation}
    B(U,V)
    =
    \left\{
        i\in D(U,V):
        s_i(U)\neq s_i(V)
    \right\}.
\end{equation}
Thus, $B(U,V)$ consists of the differing qubits on which exactly one
of $U_i,V_i$ is the identity.

\begin{definition}[{Universal gauge invariance and generic gauge dependence}]
\label{def:generic_gauge_dependence}
Let \(\mathcal{V}\) denote the \emph{ambient generator space}, namely the real vector space of Hermiticity-preserving,
trace-annihilating superoperators on \(N\) qubits, and let
\(f:\mathcal{V}\to\mathbb{R}\) be a linear functional. Define
\begin{equation}
    I_f
    =
    \left\{
        \mathcal{L}\in \mathcal{V}:
        f(\mathcal{L}_\theta)=f(\mathcal{L})
        \text{ for every }\theta\in\mathbb{R}^N
    \right\}.
\end{equation}
{
\begin{enumerate}
    \item \emph{Universal gauge invariance.} We say that \(f\) is universally gauge invariant under \(\mathcal S_\theta\) if \(I_f=\mathcal V\).
    \item \emph{Gauge dependence at an ambient generator.} We say that \(f\) is gauge dependent at \(\mathcal L\in\mathcal V\) in the ambient sense if \(\mathcal L\notin I_f\).
    \item \emph{Generic gauge dependence in the ambient sense.} We say that \(f\) is generically gauge dependent in the ambient sense if \(I_f\) is a proper linear subspace of \(\mathcal V\).
    \item \emph{Generic gauge dependence among physical Lindbladians.} For a fixed support pattern \(\mathsf P\), let \(K_{\mathsf P}=\operatorname{Lind}(\mathsf P)\). We say that \(f\) is generically gauge dependent among physical Lindbladians in \(K_{\mathsf P}\) if there exists \(O_f\subseteq K_{\mathsf P}\), open and dense relative to \(K_{\mathsf P}\), whose complement in \(K_{\mathsf P}\) has Lebesgue measure zero in \(\mathcal V_{\mathsf P}\), such that for every \(\mathcal L\in O_f\) and every neighborhood \(U\) of zero in \(\mathbb R^N\), there is a \(\theta\in U\) with \(\mathcal L_\theta\in K_{\mathsf P}\) and \(f(\mathcal L_\theta)\neq f(\mathcal L)\). When coefficient bounds $|h_c|,|\alpha_{ab}|\leq 1$ are imposed, the same definition applies to the corresponding subset of \(K_{\mathsf P}\).
\end{enumerate}
Unless explicitly qualified as ambient, \emph{generic gauge dependence} refers to the physical notion in item~4.}
\end{definition}

{
\begin{remark}
For each fixed \(\theta\), the map \(\mathcal L\mapsto f(\mathcal L_\theta)-f(\mathcal L)\) is linear. Hence \(I_f\), being the intersection of the kernels of these maps, is a linear subspace of \(\mathcal V\). If it is proper, its complement is open, dense, and of full Lebesgue measure in \(\mathcal V\). Thus every linear coefficient functional is either universally gauge invariant or generically gauge dependent in the ambient sense, and a single example of ambient gauge dependence suffices to establish the latter. These observations do not impose positivity of the GKS matrix. For the locally allowed Hamiltonian and dissipative components classified below, Corollary~\ref{cor:generic_physical_gauge_dependence} establishes the further implication from ambient generic dependence to generic dependence among physical Lindbladians.
\end{remark}
}

\begin{theorem}[{Universal gauge invariance and ambient generic gauge dependence of Pauli-basis coefficients}]
\label{thm:gauge_invariant_coefficients}
Let $\mathcal{L}$ be Hermiticity preserving and trace-annihilating.

\begin{enumerate}
    \item If $U=V$, then $\chi_{U,U}$ is real and gauge invariant.

    \item If $|D(U,V)|=1$ and $B(U,V)=\varnothing$, then
    $\Re\chi_{U,V}$ is gauge invariant. Equivalently, $U$ and $V$
    differ on exactly one qubit, and both have a nonidentity Pauli on
    that qubit.

    \item If $|D(U,V)|=1$ and $|B(U,V)|=1$, then
    $\Im\chi_{U,V}$ is gauge invariant. Equivalently, $U$ and $V$
    differ on exactly one qubit, and exactly one has the identity on
    that qubit.
\end{enumerate}

{These are exactly the universally gauge-invariant components. Every other component is generically gauge dependent in the ambient sense of Definition~\ref{def:generic_gauge_dependence}.\footnote{All genericity statements in this theorem concern the ambient space \(\mathcal V\), but in Corollary~\ref{cor:generic_physical_gauge_dependence},
we strengthen this result to generic gauge dependence among physical
Lindbladians.}} In particular, the complementary
component of $\chi_{U,V}$ is generically gauge dependent when
$|D(U,V)|=1$, and both the real and imaginary components are
generically gauge dependent when $|D(U,V)|\geq 2$.
\end{theorem}

\begin{proof}
Fix $U,V$ and a summation Pauli $R$ in
\eqref{eq:chi_gauge_transform}. Define
\begin{equation}
    G_R
    =
    \left\{
        i\in D(U,V):
        s_i(R\star U)\neq s_i(R\star V)
    \right\}.
\end{equation}
These are precisely the qubits that contribute a nontrivial factor
depending on $\theta_i$ to the summand indexed by $R$.

Consider a qubit $i\in D(U,V)\setminus B(U,V)$. In this case both $U_i$ and $V_i$
are non-identity and distinct. We then consider the local phase $\varphi_{R_i}(U_i,V_i)$ appearing in
$\varphi_R(U,V)$. If $i\in G_R$, then $s_i(R\star U)\neq s_i(R\star V)$, which implies either $R_i=U_i$ or $R_i=V_i$. Without loss of generality, we assume $R_i=U_i$. Under this assumption, $\omega(R_i,R_i\star U_i)=1$, and $\omega(R_i,R_i\star V_i)=\omega(U_i,U_i\star V_i)=\pm i$ since $U_i$ and $V_i$ are distinct non-identity Pauli matrices. Therefore $\varphi_{R_i}(U_i,V_i)=\omega(R_i,R_i\star U_i)\omega(R_i,R_i\star V_i)=\pm i$. If $i\notin G_R$, then $s_i(R\star U)=s_i(R\star V)$, which implies $R_i\neq U_i$, $R_i\neq V_i$. {So either $R_i=I$ or $R_i$ is the third non-identity Pauli that is different from $U_i,V_i$. If the former, both local phases equal \(1\). If the latter, both phases are imaginary. In either case, \(\varphi_{R_i}(U_i,V_i)=\pm1\).}
Therefore $\varphi_{R_i}(U_i,V_i)$ is imaginary if and only if $i\in G_R$. 

On the other
hand, if $i\in B(U,V)$, exactly one of $U_i,V_i$ is the identity, and a similar argument shows that
the local phase $\varphi_{R_i}(U_i,V_i)$ is imaginary exactly when $i\notin G_R$. The local
phase is real on every qubit outside $D(U,V)$. Therefore, we introduce a variable $\epsilon_R$ such that 
$\epsilon_R=0$ when $\varphi_R(U,V)$ is real and $\epsilon_R=1$ when
it is imaginary, and then
\begin{equation}
\label{eq:gauge_phase_parity}
    \epsilon_R
    \equiv
    |G_R\cap(D(U,V)\setminus B(U,V))|
    +
    |B(U,V)\setminus G_R|
    \equiv
    |B(U,V)|+|G_R|
    \pmod 2.
\end{equation}

Conversely, every subset $G\subseteq D(U,V)$ can occur as $G_R$ for
some choice of $R$. It follows that the component of phase parity
$p\in\{0,1\}$ contains a gauge-dependent summand exactly when there
exists a nonempty $G\subseteq D(U,V)$ satisfying
\begin{equation}
    |B(U,V)|+|G|\equiv p\pmod 2.
\end{equation}
Here $p=0$ corresponds to the real part and $p=1$ to the imaginary
part. Moreover, when \(G\neq\varnothing\), choose \(i\in G\) and, in
constructing \(R\), take \(R_i=U_i\). The remaining tensor factors
of \(R\) can be chosen so that \(G_R=G\). Since \(U_i\neq V_i\), this
choice ensures that \(R\neq V\), and hence \(R\star V\neq I\).
Therefore, the corresponding Pauli-transfer coefficient
\(T_{R\star V,R\star U}\) is not forced to vanish by the
trace-annihilation constraint \(T_{I,Q}=0\).

If $D(U,V)=\varnothing$, then $U=V$ and the only possible choice is
$G=\varnothing$. Hence no summand is rescaled, and
$\chi_{U,U}$ is gauge invariant. It is real because the $\chi$ matrix
of a Hermiticity-preserving map is Hermitian.

If $|D(U,V)|=1$, the unrescaled summands, corresponding to $G_R=\varnothing$, have phase parity
$|B(U,V)|$, while the rescaled summands, corresponding to $G_R=D(U,V)$, have the opposite parity.
When $B(U,V)=\varnothing$, the unrescaled summands are real, proving
that $\Re\chi_{U,V}$ is invariant. When $|B(U,V)|=1$, the unrescaled
summands are imaginary, proving that $\Im\chi_{U,V}$ is invariant.

Finally, suppose that $|D(U,V)|\geq 2$. Then $D(U,V)$ has nonempty
subsets of both parities, for example a singleton and a two-element
subset. Equation~\eqref{eq:gauge_phase_parity} therefore shows that
both the real and imaginary parts contain gauge-dependent summands.
For generic $T_{P,Q}$, these terms cannot cancel: distinct gauge
exponents are linearly independent functions of $\theta$, while
cancellations among terms with the same exponent impose proper linear
constraints on the Pauli-transfer coefficients. Thus such
cancellations occur only on a nongeneric subset of coefficient space.
\end{proof}

{
We now relate the \(\chi\)-matrix classification to the Hamiltonian
and dissipative coefficients. For nonidentity Pauli operators
\(P_a,P_b,P_c\), one has
$$
\chi_{P_a,P_b}=\alpha_{ab},
\quad
h_c=\Im\chi_{I,P_c}=-\Im\chi_{P_c,I}.
$$
The second identity holds because the anticommutator operator
\(\sum_{a,b}\alpha_{ab}P_bP_a\) is Hermitian and therefore has real
Pauli coefficients, so its contribution to \(\chi_{I,P_c}\) is real.
These identities hold throughout the ambient generator space
\(\mathcal V\), where the Kossakowski matrix is Hermitian but need
not be positive. Thus Theorem~\ref{thm:gauge_invariant_coefficients}
classifies all Hamiltonian and dissipative components. In particular,
single-qubit Hamiltonian coefficients are universally gauge invariant,
whereas higher-weight Hamiltonian coefficients are generically gauge
dependent in the ambient sense.
Corollary~\ref{cor:generic_physical_gauge_dependence} establishes
generic dependence among physical Lindbladians for every locally
allowed component classified as dependent.
}

For the Lindbladian representation in
\eqref{eq:lindbladian_defn}, whenever $U=P_a$ and $V=P_b$ are
nonidentity dissipative basis elements, one has
$
\chi_{U,V}=\alpha_{ab}
$.
Consequently, all diagonal coefficients $\alpha_{aa}$ are gauge
invariant. {We name the two classes of gauge-invariant off-diagonal
components as follows.}

\begin{definition}[{Type I Off-Diagonal Coefficients}]
    \label{defn:typei_coeff}
    {For a Lindbladian $\mathcal L = -i\sum_c h_c[P_c, \cdot] + \sum_{a, b}\alpha_{ab}\mathcal D_{ab}$, we say that a component $\Re\alpha_{ab}$ is Type I, or equivalently, $(a, b) \in \mathcal T_1$, if $|D(P_a, P_b)| = 1$ and $B(P_a, P_b) = \varnothing$. In other words, $P_a, P_b$ satisfy Case 2 of Proposition~\ref{thm:gauge_invariant_coefficients}.}
\end{definition}

\begin{definition}[{Type II Components}]
    \label{defn:typeii_coeff}
    {For a Lindbladian $\mathcal L = -i\sum_c h_c[P_c,\cdot]+\sum_{a,b}\alpha_{ab}\mathcal D_{ab}$, we say that a component $\Im\alpha_{ab}$ is Type II, or equivalently, $(a,b)\in\mathcal T_2$, if
    $|D(P_a,P_b)|=|B(P_a,P_b)|=1$. Thus, $P_a$ and $P_b$ differ on a unique qubit, and exactly one of them is the identity on that qubit.  We also include every single-qubit Hamiltonian coefficient $h_c$ among the Type~II components, i.e., $c\in\mathcal{T}_2$ if $|\mathrm{supp}(P_c)|=1$.}
\end{definition}

{This also classifies every Hamiltonian coefficient. For every nonidentity Pauli \(P_c\), one has \(h_c=\Im\chi_{I,P_c}\): the anticommutator contribution to \(\chi_{I,P_c}\) is real because \(\sum_{a,b}\alpha_{ab}P_bP_a\) is Hermitian, as can also be seen from the expansion below. The same identity holds throughout the ambient generator space \(\mathcal V\), where the GKS matrix is Hermitian but need not be positive. Theorem~\ref{thm:gauge_invariant_coefficients} therefore shows that single-qubit Hamiltonian coefficients are universally gauge-invariant Type~II components, whereas every higher-weight Hamiltonian coefficient is generically gauge dependent in the ambient sense. Together with \(\chi_{P_a,P_b}=\alpha_{ab}\), this gives the ambient classification of all Hamiltonian and dissipative components. Corollary~\ref{cor:generic_physical_gauge_dependence} upgrades the dependence of every locally allowed component in this classification to generic dependence among physical Lindbladians.}

We note that {Theorem}~\ref{thm:gauge_invariant_coefficients} does not consider whether the transformed generator is still a Lindbladian. In other words, for certain values of $\theta$, it may be that the corresponding $\mathcal{L}_{\theta}$ is not a valid Lindbladian, as the gauge transformation $\mathcal{S}_\theta$ is not guaranteed to preserve the positivity of the GKS matrix. This leads to a natural question: can the positivity of the GKS matrix eliminate gauge degrees of freedom? {Below we show that, for almost every physical Lindbladian with a prescribed support pattern, all locally allowed components classified as dependent in the ambient sense can still vary under arbitrarily small gauge transformations that preserve physicality (Corollary~\ref{cor:generic_physical_gauge_dependence}).}

For \(U,V\in\bar{\mathcal{P}}_N\) and \(p\in\{0,1\}\), define the
real linear functional
\begin{equation}
\label{eq:pauli_component_functional}
    f_{U,V,p}(\mathcal{L})
    =
    \begin{cases}
        \Re\chi_{U,V}(\mathcal{L}), & p=0,\\
        \Im\chi_{U,V}(\mathcal{L}), & p=1.
    \end{cases}
\end{equation}
These functionals can be related explicitly to the Hamiltonian coefficients and the GKS matrix elements by expanding
\begin{equation*}
\begin{aligned}
    \mathcal{L}[\rho]
    ={}&
    \sum_{(a,b)\in\mathcal{I}_d}
    \alpha_{ab}P_a\rho P_b
    -i\sum_{c\in\mathcal{I}_h}h_cP_c\rho
    +i\sum_{c\in\mathcal{I}_h}h_c\rho P_c
    -
    \frac{1}{2}
    \sum_{(a,b)\in\mathcal{I}_d}
    \alpha_{ab}{P_bP_a}\rho
    -
    \frac{1}{2}
    \sum_{(a,b)\in\mathcal{I}_d}
    \alpha_{ab}\rho{P_bP_a} .
\end{aligned}
\end{equation*}
If \(U=P_a\) and \(V=P_b\) are both nonidentity Pauli operators, then
\(\chi_{U,V}=\alpha_{ab}\), so \(f_{U,V,0}\) and \(f_{U,V,1}\)
extract \(\Re\alpha_{ab}\) and \(\Im\alpha_{ab}\), respectively. If
\(U\neq I\), the one-identity coefficients are
\begin{equation*}
\begin{aligned}
    \chi_{U,I}
    =
    -ih_U
    -
    \frac{1}{2}
    \sum_{\substack{(a,b)\in\mathcal{I}_d\\P_a\star P_b=U}}
    {\omega(P_b,P_a)}\alpha_{ab},
    \quad
    \chi_{I,U}
    =
    ih_U
    -
    \frac{1}{2}
    \sum_{\substack{(a,b)\in\mathcal{I}_d\\P_a\star P_b=U}}
    {\omega(P_b,P_a)}\alpha_{ab},
\end{aligned}
\end{equation*}
where \(h_U=h_c\) if \(U=P_c\) for some \(c\in\mathcal{I}_h\), and
\(h_U=0\) otherwise. Thus \(f_{U,I,p}\) and \(f_{I,U,p}\) extract the
real or imaginary parts of the coefficients multiplying \(U\rho\)
and \(\rho U\), including the contributions from terms of the form
\({P_bP_a}\rho\) and \(\rho{P_bP_a}\). In particular, the Hamiltonian
coefficient is isolated by
$
    h_U
    =
    \frac{\chi_{I,U}-\chi_{U,I}}{2i}.
$

For \(\varsigma=(U,V,p)\), we write
\begin{equation*}
    f_\varsigma=f_{U,V,p},
    \quad
    S_\varsigma=\supp(U)\cup\supp(V).
\end{equation*}
Given a support pattern \(\mathsf{P}\), let
\begin{equation}
\label{eq:dependent_components_support_pattern}
    \operatorname{Dep}(\mathsf{P})
    =
    \left\{
        \varsigma=(U,V,p):
        {\begin{aligned}
        &f_\varsigma\text{ is generically gauge dependent in the ambient sense}\\
        &\text{and }S_\varsigma\subseteq C
        \text{ for some }C\in\mathsf{P}
        \end{aligned}}
    \right\}.
\end{equation}
{The set \(\operatorname{Dep}(\mathsf P)\) is defined using the ambient notion in Definition~\ref{def:generic_gauge_dependence} and the classification in Theorem~\ref{thm:gauge_invariant_coefficients}. Corollary~\ref{cor:generic_physical_gauge_dependence} below establishes generic gauge dependence among physical Lindbladians for every component in this set.}

\begin{proposition}[Simultaneous variation within a support pattern]
\label{prop:local_simultaneous_gauge_dependence}
For every support pattern \(\mathsf{P}\), there exists a nonempty
subset
$
    O(\mathsf{P})
    \subseteq
    \operatorname{Lind}(\mathsf{P})
$
that is open relative to \(\mathcal{V}_{\mathsf{P}}\) and has the
following property. For every
\(\mathcal{L}\in O(\mathsf{P})\) and every neighborhood \(U\) of
\(0\) in \(\mathbb{R}^N\), there exists \(\theta\in U\) such that
$
    \mathcal{L}_\theta
    \in
    \operatorname{Lind}(\mathsf{P})
$
and
$
    f_\varsigma(\mathcal{L}_\theta)
    \neq
    f_\varsigma(\mathcal{L})$
for every $\varsigma\in\operatorname{Dep}(\mathsf{P})$.
Thus, all locality-allowed {components that are generically gauge dependent in the ambient sense} 
{can be changed simultaneously by an arbitrarily small gauge transformation while preserving both physicality and the prescribed support pattern.}
\end{proposition}

\begin{proof}
We will prove this proposition by first showing that superoperators in $\mathcal{V}_{\mathsf{P}}$ with a $\varsigma$-component invariant under $\mathcal{S}_\theta$ are contained in a finite union of proper subspaces of $\mathcal{V}_{\mathsf{P}}$. We then construct a set of valid Lindbladians in $\mathcal{V}_\mathsf{P}$ that is open relative to $\mathcal{V}_{\mathsf{P}}$. Removing elements of the finite union of subspaces from the open set we construct, we then end up with an open set of valid Lindbladians. For Lindbladians within this set, we will then show that $\theta$ changes all parameters except for rare cases that form a set with empty interior, thereby proving the existence of a Lindbladian $\mathcal{L}_\theta$ whose every component is different from $\mathcal{L}$. Below we will proceed step by step.

Because \(\mathcal{S}_\theta\) is a product of single-qubit maps, for
every \(C\in\mathsf{P}\) and \(\mathcal{K}_C\in\mathcal{V}_C\),
\begin{equation}
\label{eq:gauge_preserves_locality}
    \bigl(\iota_C(\mathcal{K}_C)\bigr)_\theta
    =
    \iota_C\bigl((\mathcal{K}_C)_{\theta_C}\bigr),
\end{equation}
where \(\theta_C\) is the restriction of \(\theta\) to \(C\).
Consequently, \(\mathcal{V}_{\mathsf{P}}\) is invariant under the
gauge transformation.

For \(\varsigma\in\operatorname{Dep}(\mathsf{P})\), we consider the set of superoperators with support pattern $\mathsf{P}$ whose component extracted by $f_\varsigma$ remains invariant under the gauge transformation:
\begin{equation*}
    J_\varsigma(\mathsf{P})
    =
    I_{f_\varsigma}\cap\mathcal{V}_{\mathsf{P}},
\end{equation*}
where \(I_{f_\varsigma}\) is defined in
Definition~\ref{def:generic_gauge_dependence}. This is a linear
subspace of \(\mathcal{V}_{\mathsf{P}}\). Below we will prove that it is a
proper subspace.

Choose \(C\in\mathsf{P}\) such that \(S_\varsigma\subseteq C\), and regard
\(U\) and \(V\) as Pauli operators on \(C\). Applying
Theorem~\ref{thm:gauge_invariant_coefficients} to the
\(|C|\)-qubit system shows that the corresponding local functional
is {generically gauge dependent in the ambient sense} on \(\mathcal{V}_C\). Hence there are
\(\mathcal{K}_C\in\mathcal{V}_C\) and a local gauge vector
\(\vartheta\) such that
\begin{equation*}
    f_\varsigma\bigl(\iota_C((\mathcal{K}_C)_\vartheta)\bigr)
    \neq
    f_\varsigma\bigl(\iota_C(\mathcal{K}_C)\bigr).
\end{equation*}
Here the local and embedded \(\chi\)-representation coefficients
agree. Extending \(\vartheta\) arbitrarily to a gauge vector on the
full system and applying
\eqref{eq:gauge_preserves_locality} shows that
\(\iota_C(\mathcal{K}_C)\notin J_\varsigma(\mathsf{P})\). Therefore
\(J_\varsigma(\mathsf{P})\) is proper.

We next construct a relatively open set of valid Lindbladians in
\(\mathcal{V}_{\mathsf{P}}\). For each \(C\in\mathsf{P}\), choose
a traceless Hermitian basis \(\{F_a^{(C)}\}_a\). Every
\(\mathcal{K}_C\in\mathcal{V}_C\) admits a GKS representation
\begin{equation*}
\begin{aligned}
    \mathcal{K}_C[\rho_C]
    =
    -i[H_C,\rho_C]
    +
    \sum_{a,b}(A_C)_{ab}
    \left(
        F_a^{(C)}\rho_C F_b^{(C)}
        -
        \frac{1}{2}
        \{F_b^{(C)}F_a^{(C)},\rho_C\}
    \right),
\end{aligned}
\end{equation*}
where \(H_C\) and \(A_C\) are Hermitian. The linear map
\begin{equation*}
    \Phi:
    (H_C,A_C)_{C\in\mathsf{P}}
    \longmapsto
    \sum_{C\in\mathsf{P}}\iota_C(\mathcal{K}_C)
\end{equation*}
is surjective onto \(\mathcal{V}_{\mathsf{P}}\). Since a
surjective linear map between finite-dimensional vector spaces is
open, the set
\begin{equation*}
    O_0(\mathsf{P})
    =
    \Phi
    \left(
        \left\{
            (H_C,A_C)_{C\in\mathsf{P}}:
            A_C\succ0
            \text{ for every }C\in\mathsf{P}
        \right\}
    \right)
\end{equation*}
is nonempty and open relative to
\(\mathcal{V}_{\mathsf{P}}\). Moreover,
\begin{equation*}
    O_0(\mathsf{P})
    \subseteq
    \operatorname{Lind}(\mathsf{P}),
\end{equation*}
because every element of \(O_0(\mathsf{P})\) is a sum of embedded
local Lindbladians.

{To impose the coefficient normalization, intersect \(O_0(\mathsf P)\) with the open set where every Hamiltonian and dissipative coefficient has magnitude less than \(1\). This intersection is nonempty because a sufficiently small positive rescaling of any element of \(O_0(\mathsf P)\) remains in \(O_0(\mathsf P)\). We use this intersection as \(O_0(\mathsf P)\) below.}

There are only finitely many elements of
\(\operatorname{Dep}(\mathsf{P})\), and each
\(J_\varsigma(\mathsf{P})\) is a proper linear subspace. Therefore,
\begin{equation*}
    O(\mathsf{P})
    =
    O_0(\mathsf{P})
    \setminus
    \bigcup_{\varsigma\in\operatorname{Dep}(\mathsf{P})}
    J_\varsigma(\mathsf{P})
\end{equation*}
is nonempty and open relative to
\(\mathcal{V}_{\mathsf{P}}\).

Fix \(\mathcal{L}\in O(\mathsf{P})\), and define
\begin{equation*}
    \Delta_\varsigma(\theta)
    =
    f_\varsigma(\mathcal{L}_\theta)-f_\varsigma(\mathcal{L}).
\end{equation*}
For every \(\varsigma\in\operatorname{Dep}(\mathsf{P})\),
\(\Delta_\varsigma\) is a nonzero real-analytic function. The zero set of a
nonzero real-analytic function has empty interior, and hence so does
the finite union
\begin{equation*}
    \bigcup_{\varsigma\in\operatorname{Dep}(\mathsf{P})}
    \{\theta:\Delta_\varsigma(\theta)=0\}.
\end{equation*}
Furthermore, \(O_0(\mathsf{P})\) is relatively open,
\(\mathcal{L}_0=\mathcal{L}\), and the gauge action is continuous.
Thus \(\mathcal{L}_\theta\in O_0(\mathsf{P})\) for all sufficiently
small \(\theta\). Every neighborhood of \(0\) therefore contains a
\(\theta\) for which \(\mathcal{L}_\theta\) is a Lindbladian
supported on \(\mathsf{P}\) and every \(\Delta_\varsigma(\theta)\) is
nonzero.
\end{proof}

{
\begin{corollary}[Generic gauge dependence among physical Lindbladians]
\label{cor:generic_physical_gauge_dependence}
For every support pattern \(\mathsf P\), the set \(O(\mathsf P)\) in
Proposition~\ref{prop:local_simultaneous_gauge_dependence} can be chosen
open in \(\mathcal V_{\mathsf P}\) and dense in
\(\operatorname{Lind}(\mathsf P)\), with
\(\operatorname{Lind}(\mathsf P)\setminus O(\mathsf P)\) of Lebesgue
measure zero in \(\mathcal V_{\mathsf P}\). Thus, for almost every
physical Lindbladian with this support pattern, all components in
\(\operatorname{Dep}(\mathsf P)\) can be changed simultaneously by an
arbitrarily small gauge transformation preserving physicality and the
support pattern. This establishes generic gauge dependence in the physical
sense of Definition~\ref{def:generic_gauge_dependence} for every component
in \(\operatorname{Dep}(\mathsf P)\).
The same conclusion holds within the subset satisfying
the coefficient bounds \(|h_c|,|\alpha_{ab}|\leq1\).
\end{corollary}
}

\begin{proof}
{
Write \(V=\mathcal V_{\mathsf P}\) and
\(K=\operatorname{Lind}(\mathsf P)\). The set \(K\) is closed and
convex, since the Kossakowski matrix depends linearly on the generator
and physicality is its positive semidefiniteness. The preceding proof
shows that \(K\) has nonempty interior in \(V\) and that each
\(J_\varsigma(\mathsf P)\) is a proper linear subspace of \(V\). Set
\begin{equation*}
    O(\mathsf P)
    =\operatorname{int}_{V}K
    \setminus\bigcup_{\varsigma\in\operatorname{Dep}(\mathsf P)}
    J_\varsigma(\mathsf P).
\end{equation*}
The interior of a full-dimensional convex set is dense in that set,
and its boundary has Lebesgue measure zero. Removing finitely many
proper linear subspaces therefore gives the asserted openness,
density, and full measure.
}

{
For \(\mathcal L\in O(\mathsf P)\), continuity of the gauge action
and its preservation of \(V\) imply \(\mathcal L_\theta\in K\) for
all sufficiently small \(\theta\). Each function
\(\theta\mapsto f_\varsigma(\mathcal L_\theta)-f_\varsigma(\mathcal L)\)
is real analytic by Eq.~\eqref{eq:ptm_gauge_transform} and is not identically
zero because \(\mathcal L\notin J_\varsigma(\mathsf P)\). As in the preceding proof,
the finite union over \(\varsigma\in\operatorname{Dep}(\mathsf P)\) of their zero sets has empty interior, so every
neighborhood of zero contains a \(\theta\) that changes all these
components while keeping \(\mathcal L_\theta\in K\).
For the coefficient-bounded version, apply the
same argument to the intersection of \(K\) with the coefficient bounds;
this set is again closed and convex with nonempty interior, as follows
by scaling an interior point of \(K\) to satisfy all bounds strictly.
}
\end{proof}

{For fixed locality, the gauge vectors in
Proposition~\ref{prop:local_simultaneous_gauge_dependence} and
Corollary~\ref{cor:generic_physical_gauge_dependence} can be chosen
in \((-1,1)^N\). Remark~\ref{rem:admissible_spam_gauge} therefore
supplies admissible SPAM models making the original and transformed
Lindbladians experimentally indistinguishable. These generic
coefficient ambiguities thus obstruct uniform learning over the
allowed SPAM class.}

We note that the preceding proposition is qualitative and does not require
bounded patch size $\mathsf{k}$ or bounded overlap $\mathsf{d}$. 
{The proposition and corollary establish generic gauge dependence among physical Lindbladians, but do not quantify the size of the ambiguity.} Below we will show that there exists a Lindbladian for which such changes can be lower bounded by a constant that depends only on  $\mathsf{k}$ and $\mathsf{d}$, and not on the system size $N$. {For the Lindbladian construction below and for each patch $C$ in the support pattern, we take the candidate lists of Lindbladian terms to include every nonidentity Hamiltonian Pauli supported within \(C\) and every pair of nonidentity Pauli operators satisfying \(\supp(P_a)\cup\supp(P_b)\subseteq C\), with any listed coefficient allowed to vanish.}

\begin{theorem}[System-size-independent marginal gauge width]
\label{thm:uniform_gauge_width}
For every \(\mathsf{k}\geq1\) and \(\mathsf{d}\geq0\), there exists
a constant
$
    \eta_{\mathsf{k},\mathsf{d}}>0
$
with the following property. For every \(N\) and every
\((\mathsf{k},\mathsf{d})\)-bounded degree local support pattern
\(\mathsf{P}\) on \([N]\), there exists
$
    \mathcal{L}_*
    \in
    \operatorname{Lind}(\mathsf{P})
$
such that, for
$
    \Theta_N=[-1,1]^N,
$
one has
$
    \mathcal{L}_{*,\theta}
    \in
    \operatorname{Lind}(\mathsf{P})
$
for every $\theta\in\Theta_N$. {The generators can be chosen so that all their Hamiltonian and dissipative coefficients have magnitude at most \(1\), uniformly over \(\theta\in\Theta_N\).}
Moreover, every \(\varsigma\in\operatorname{Dep}(\mathsf{P})\) satisfies
\begin{equation}
\label{eq:uniform_gauge_oscillation}
    \max_{\theta\in\Theta_N}
    f_\varsigma(\mathcal{L}_{*,\theta})
    -
    \min_{\theta\in\Theta_N}
    f_\varsigma(\mathcal{L}_{*,\theta})
    \geq
    \eta_{\mathsf{k},\mathsf{d}}.
\end{equation}
The constant is independent of both \(N\) and the particular support
pattern, although \(\mathcal{L}_*\) may depend on the support pattern.
{With the SPAM channels of
Remark~\ref{rem:admissible_spam_gauge}, the entire family satisfies
Definition~\ref{defn:spam_noise_model} and gives identical outcome
distributions for every allowed controlled experiment.}
\end{theorem}

\begin{proof}
For each \(1\leq n\leq\mathsf{k}\), apply
Proposition~\ref{prop:local_simultaneous_gauge_dependence} to the
support pattern with a single patch \(\{[n]\}\), and choose an \(n\)-qubit
Lindbladian \(\mathcal{K}_n\) at which every generically
gauge-dependent \(n\)-qubit Pauli-basis component is gauge
dependent. Let
$
    \operatorname{Dep}_n
    =
    \operatorname{Dep}(\{[n]\})
$
denote the finite collection of these components. For component label
\(\varsigma\in\operatorname{Dep}_n\), define
\begin{equation}
\label{eq:range_params_Kn}
    g_{n,\varsigma}(\vartheta)
    =
    f_\varsigma((\mathcal{K}_n)_\vartheta),
    \quad
    \omega_{n,\varsigma}
    =
    \max_{\vartheta\in[-1,1]^n}g_{n,\varsigma}(\vartheta)
    -
    \min_{\vartheta\in[-1,1]^n}g_{n,\varsigma}(\vartheta).
\end{equation}
$\omega_{n,\varsigma}$ therefore describes how much the $\varsigma$ component can change when gauge transformation $\mathcal{S}_\vartheta$ is applied to $n$ qubits. It may be that the transformed $(\mathcal{K}_n)_{\vartheta}$ is no longer a valid Lindbladian, which is acceptable in our construction.
Each \(g_{n,\varsigma}\) is real analytic and nonconstant, so
\(\omega_{n,\varsigma}>0\). Since there are only finitely many pairs
\((n,\varsigma)\), the constants
\begin{equation}
\label{eq:local_oscillation_constants}
\begin{aligned}
    \mu
    &=
    \min_{\substack{1\leq n\leq\mathsf{k}\\
                    \varsigma\in\operatorname{Dep}_n}}
    \omega_{n,\varsigma}
    >0,\quad 
    \Omega
    =
    \max_{\substack{1\leq n\leq\mathsf{k}\\
                    \varsigma\in\operatorname{Dep}_n}}
    \omega_{n,\varsigma}
    <\infty
\end{aligned}
\end{equation}
depend only on \(\mathsf{k}\).

We will construct $\mathcal{L}_*$ as follows. First, because the intersection graph of \(\mathsf{P}\) has degree at most
\(\mathsf{d}\), we can choose a proper coloring
\begin{equation*}
    \kappa:
    \mathsf{P}
    \longrightarrow
    \{0,1,\ldots,\mathsf{d}\}.
\end{equation*}
Fix an ordering of the qubits in every patch \(C\), and let
\(\mathcal{K}_C\) be the corresponding relabeled copy of
\(\mathcal{K}_{n}\) that we have constructed above for $n=|C|$. We then define the complete Pauli-depolarizing
generator on \(C\) by
\begin{equation*}
    \mathcal{G}_C[\rho_C]
    =
    \sum_{R\in\bar{\mathcal{P}}_{|C|}\setminus\{I\}}
    (R\rho_C R-\rho_C).
\end{equation*}
This enables us to define $\mathcal{L}_*$:
\begin{equation}
\label{eq:uniform_width_generator}
    \mathcal{L}_*
    =
    \sum_{C\in\mathsf{P}}
    \iota_C
    \left(
        \mathcal{G}_C
        +
        \varepsilon\zeta^{\kappa(C)}\mathcal{K}_C
    \right),
\end{equation}
where $\varepsilon$ and $\zeta$ are parameters that we will choose to ensure $\mathcal{L}_{*,\theta}\in \mathrm{Lind}(\mathsf{P})$ for all $\theta\in\Theta_N$, and to ensure the component change contributed by each $\mathcal{K}_C$ is not canceled by those from other overlapping patches $C'\neq C$ (more precisely, this is why we choose the exponential scaling with the color index $\zeta^{\kappa(C)}$). The specific choices we use are
\begin{equation}
\label{eq:uniform_buffer_constants}
    \varepsilon
    =
    \frac{1}{2\max\{1,M\}},
    \quad
    \zeta
    =
    \min\left\{
        \frac{1}{2},
        \frac{\mu}{4\Omega}
    \right\}.
\end{equation}
where we let \(A_n(\vartheta)\) denote the Hermitian dissipative GKS matrix of
\((\mathcal{K}_n)_\vartheta\) in the nonidentity Pauli basis, and define
\[
    M=\max_{\substack{1\leq n\leq\mathsf{k}\\
                    \vartheta\in[-1,1]^n}}
    \|A_n(\vartheta)\|_{\mathrm{op}}.
\]
$M$ is finite due to the compactness of $[-1,1]^n$.

We will first show that our choice of $\varepsilon$ ensures $\mathcal{L}_{*,\theta}$ remains a valid Lindbladian. For this we only need the GKS matrix of each summand in \eqref{eq:uniform_width_generator}, i.e., $\mathcal{G}_C + \varepsilon\zeta^{\kappa(C)}\mathcal{K}_C$, to be positive semidefinite after applying the gauge transformation with parameter $\theta_C\in[-1,1]^{|C|}$.
For every \(\theta\in\Theta_N\), the GKS matrix of the transformed
summand
\[
\left(
    \mathcal{G}_C
    +
    \varepsilon\zeta^{\kappa(C)}\mathcal{K}_C
\right)_{\theta_C} = \mathcal{S}_{\theta_C}\circ\left(
    \mathcal{G}_C
    +
    \varepsilon\zeta^{\kappa(C)}\mathcal{K}_C
\right)\circ\mathcal{S}_{\theta_C}^{-1}
\]
on \(C\) is
\begin{equation*}
\begin{aligned}
    I_{4^{|C|}-1}
    +
    \varepsilon\zeta^{\kappa(C)}
    A_{|C|}(\theta_C)
    \succeq
    (1-\varepsilon M)I_{4^{|C|}-1}
    \succeq
    \frac{1}{2}I_{4^{|C|}-1}.
\end{aligned}
\end{equation*}
Here we have used the fact that $\mathcal{G}_C$ is invariant under the gauge transformation and has identity GKS matrix.
Every transformed patch summand
$
\left(
    \mathcal{G}_C
    +
    \varepsilon\zeta^{\kappa(C)}\mathcal{K}_C
\right)_{\theta_C}
$
is therefore a Lindbladian, and hence
\begin{equation*}
    \mathcal{L}_{*,\theta}
    =
    \sum_{C\in\mathsf{P}}
    \iota_C
    \left(
        \left(
            \mathcal{G}_C
            +
            \varepsilon\zeta^{\kappa(C)}\mathcal{K}_C
        \right)_{\theta_C}
    \right)
    \in
    \operatorname{Lind}(\mathsf{P}),
    \quad
    \text{for every }\theta\in\Theta_N.
\end{equation*}

Next, we will show that our choice of $\zeta$ in \eqref{eq:uniform_buffer_constants} ensures all Lindbladian components $f_{\varsigma}(\mathcal{L}_{*,\theta})$ are changed by a minimum amount that depends only on $\mathsf{k},\mathsf{d}$ when $\theta$ varies within range $[-1,1]^N$. 

 Fix
\(\varsigma=(U,V,p)\in\operatorname{Dep}(\mathsf{P})\), and let $\mathsf{P}(\varsigma)$ be the set of patches that contain $S_\varsigma=\supp(U)\cup\supp(V)$, i.e.,
$
    \mathsf{P}(\varsigma)
    =
    \{C\in\mathsf{P}:S_\varsigma\subseteq C\}.
$
For \(C\in\mathsf{P}(\varsigma)\), let \(\varsigma|_C\) denote the component
obtained by restricting \(U\) and \(V\) to \(C\) and expressing them
in the fixed ordering of that patch. Then because both $U$ and $V$ are supported entirely within $C$, we have
$
\varsigma|_C\in\operatorname{Dep}_{|C|}.
$

Because \(\varsigma\) is gauge dependent, \(S_\varsigma\neq\varnothing\). Hence all
patches in \(\mathsf{P}(\varsigma)\) pairwise intersect and consequently
receive distinct colors. Choose \(C_0\in\mathsf{P}(\varsigma)\) with
smallest color index
$
j_0=\kappa(C_0).
$
Choose
\(\vartheta^+,\vartheta^-\in[-1,1]^{|C_0|}\) attaining the maximum
and minimum, respectively, in the definition of
\(\omega_{|C_0|,\varsigma|_{C_0}}\) (\eqref{eq:range_params_Kn}). Extend them to
\(\theta^+,\theta^-\in\Theta_N\) by setting all coordinates outside
\(C_0\) equal to zero.

Only patches $C$ in \(\mathsf{P}(\varsigma)\) can contribute to the change in
\(f_\varsigma\), as otherwise $S_{\varsigma}\not\subseteq C$ and $U\cdot V$ cannot appear as a term in $\mathcal{K}_C$. The depolarizing terms $\mathcal{G}_C$ contribute no change because they are invariant under the gauge transformation. The contribution from
\(C_0\) therefore changes by at least
\(\varepsilon\zeta^{j_0}\mu\) when $\theta$ changes between $\theta^+$ and $\theta^-$ (by \eqref{eq:local_oscillation_constants}), whereas the magnitude of the change
from any other patch \(C\in\mathsf{P}(\varsigma)\) is at most
\(\varepsilon\zeta^{\kappa(C)}\Omega\) (also by \eqref{eq:local_oscillation_constants}). Since the patches
\(C\in \mathsf{P}(\varsigma)\) have distinct colors and all colors other than
\(j_0\) are larger than \(j_0\),
\begin{equation}
\begin{aligned}
    \left|
        f_\varsigma(\mathcal{L}_{*,\theta^+})
        -
        f_\varsigma(\mathcal{L}_{*,\theta^-})
    \right|
    \geq
    \varepsilon\zeta^{j_0}\mu
    -
    \varepsilon\Omega
    \sum_{\substack{C\in\mathsf{P}(\varsigma)\\C\neq C_0}}
    \zeta^{\kappa(C)}
    \geq
    \varepsilon\zeta^{j_0}
    \left(
        \mu-\frac{\Omega\zeta}{1-\zeta}
    \right)
    \geq
    \frac{\varepsilon\mu}{2}\zeta^{j_0}
    \geq
    \frac{\varepsilon\mu}{2}\zeta^{\mathsf{d}}.
\end{aligned}
\end{equation}
It follows that \eqref{eq:uniform_gauge_oscillation} holds with
\begin{equation}
\label{eq:uniform_gauge_width_constant}
    \eta_{\mathsf{k},\mathsf{d}}
    =
    \frac{\varepsilon\mu}{2}\zeta^{\mathsf{d}}
    >0,
\end{equation}
which depends only on $\mathsf{k},\mathsf{d}$.

{Finally, we normalize our construction so that all Lindbladian coefficients are bounded by $1$. let \(M_h\) be the maximum magnitude of a Hamiltonian coefficient of \((\mathcal K_n)_\vartheta\), over \(1\leq n\leq\mathsf k\) and \(\vartheta\in[-1,1]^n\). Compactness ensures \(M_h<\infty\) and is independent of $N$. Every transformed patch summand has dissipative coefficients bounded by \(1+\varepsilon M\) and Hamiltonian coefficients bounded by \(\varepsilon M_h\). Each Hamiltonian or dissipative  coefficient receives contributions from at most \(\mathsf d+1\) patches. Therefore, multiplying both \(\mathcal L_*\) and \(\eta_{\mathsf k,\mathsf d}\) by
$$
\beta_{\mathsf k,\mathsf d}
=\frac{1}{(\mathsf d+1)\max\{1+\varepsilon M,\varepsilon M_h\}}
$$
ensures that every coefficient of every \(\mathcal L_{*,\theta}\) has magnitude at most \(1\). Positivity and the support pattern are preserved, and the rescaled width still depends only on \(\mathsf k,\mathsf d\).}
\end{proof}

{For each component in \(\operatorname{Dep}(\mathsf P)\),
combining Eq.~\eqref{eq:uniform_gauge_oscillation} with
Remark~\ref{rem:admissible_spam_gauge} gives two
experimentally indistinguishable models whose values differ by at
least \(\eta_{\mathsf k,\mathsf d}\). No estimator can guarantee
additive error \(\epsilon<\eta_{\mathsf k,\mathsf d}/2\) with
success probability greater than \(1/2\) uniformly over this
family, regardless of the number of experiments: the two accuracy
intervals are disjoint, while the data distributions are identical.}

{\section{Efficient learning of universally gauge-invariant components}}

\begin{definition}[Interaction-covering partition family]
    We will partition the $N$-qubit system in $\mathsf{r}$ different ways, and denote the partitions by $\Pi_1,\Pi_2,\cdots,\Pi_\mathsf{r}$, which together we call an \emph{interaction-covering partition family}. Each partition $\Pi_i=(C_{i,1},C_{i,2},\cdots,C_{i,m_i})$, where each $C_{i,j}\subset[N]$ is called a \emph{cluster}, and $C_{i,j}\cap C_{i,k}=\varnothing$ {for distinct \(j,k\). We require \(|C_{i,j}|\leq\mathsf k\) and \(\bigcup_j C_{i,j}=[N]\)}.
    These partitions satisfy, for each $\alpha\in V_{\mathrm{int}}$, there exists $i,j$ such that $S_\alpha \subset C_{i,j}$.
\end{definition}

\begin{fact}
    There exists $\mathsf{r}=\Or(1)$ such that an interaction-covering partition family $\Pi_1,\Pi_2,\cdots,\Pi_\mathsf{r}$ exists for the Lindbladian $\mathcal{L}$ if it is $(\mathsf{k},\mathsf{d})$-bounded degree local. In particular, $\mathsf{r}\leq \mathsf{d}+1$. Each cluster, as required in the above definition, is of size at most $\mathsf{k}$.
\end{fact}

\begin{proof}
    {The known interaction graph has a proper coloring with at most \(\mathsf d+1\) colors. For each color, use its mutually disjoint term supports as clusters and add singleton clusters for uncovered qubits. Every cluster has size at most \(\mathsf k\), and the resulting family covers every term support with \(\mathsf r\leq\mathsf d+1\) partitions.} 
\end{proof}

\subsection{The diagonal part}

{Write \(\mathcal I_{\mathrm{diag}}=\{a:(a,a)\in\mathcal I_d\}\) and \(S_a=\supp(P_a)\).} We will first focus on the diagonal part of the Lindbladian
\begin{equation}
    \label{eq:diag_Lindblad}
    \mathcal{L}^{\mathrm{diag}}[\rho] = \sum_a \alpha_{aa}\left( P_{a} \rho P_{a} - \rho\right).
\end{equation}
The diagonal part, being a Pauli-Lindbladian, transforms Pauli operators in a particularly simple way
\begin{equation}
    \label{eq:diag_Lindblad_on_Pauli}
    (\mathcal{L}^{\mathrm{diag}})^{\dag}[P] = \sum_a \alpha_{aa}\left((-1)^{\braket{P,P_a}}-1\right)P.
\end{equation}

\begin{definition}[Restricted Pauli-Lindbladian]
\label{defn:restricted_pauli_lindbladian}
    For any $C\subset [N]$, we define the corresponding \emph{restricted Pauli-Lindbladian} to be 
    \[
    \mathcal{L}_C^{\mathrm{diag}}[\rho] = \sum_a \alpha_{aa}\left( P_{a,C} \rho P_{a,C} - \rho\right).
    \]
    Here $P_{a,C}$ is defined as follows: if $P_a = \Pi_{i=1}^N P_{a,i}$, where each $P_{a,i}$ is the corresponding Pauli operator acting on qubit $i$, then $P_{a,C} = \Pi_{i\in C} P_{a,i}$. 
\end{definition}

We note that the coefficients of the restricted Pauli-Lindbladian $\mathcal{L}_C^{\mathrm{diag}}$ reflect the coefficients of the original Lindbladian. More precisely, we can write the restricted Pauli-Lindbladian as
\begin{equation}
    \mathcal{L}_C^{\mathrm{diag}}[\rho] = \sum_{P:\supp(P)\subset C}\gamma_{P,C} (P\rho P-\rho).
\end{equation}

The restricted Pauli-Lindbladian $\mathcal{L}_C^{\mathrm{diag}}$ faithfully reproduces the action of the original diagonal part $\mathcal{L}^{\mathrm{diag}}$, as stated in the following lemma:
\begin{lemma}
\label{lem:restricted_pauli_lindbladian}
    For any $C\subset [N]$, if $Q$ is a Pauli operator supported on $C$, then
    \begin{equation}
        (\mathcal{L}^{\mathrm{diag}})^\dag [Q]=(\mathcal{L}_C^{\mathrm{diag}})^\dag [Q] =  \sum_a \alpha_{aa}\left((-1)^{\braket{Q,P_{a,C}}}-1\right)Q.
    \end{equation}
\end{lemma}

\subsubsection{Recovering coefficients from \texorpdfstring{$\{\gamma_{P,C_{i,j}}\}$}{gamma sub {P,C {i,j}}}}
Because of Lemma~\ref{lem:restricted_pauli_lindbladian}, we can obtain equations for $\alpha_{aa}$ if we have estimates for $\gamma_{P,C}$:
\begin{equation}
\label{eq:from_alpha_aa_to_gamma_PC}
    \gamma_{P,C} = \sum_{{a\in\mathcal I_{\mathrm{diag}}}: P_{a,C}=P} \alpha_{aa}. 
\end{equation}
This allows us to recover all $\alpha_{aa}$ through a back-substitution over the partial order given by support inclusion. Process the labels {\(a\in\mathcal I_{\mathrm{diag}}\)} in any order such that $b$ is processed before $a$ whenever $S_a\subsetneqq S_b$. When processing $a$, {choose a cluster \(C(a)=C_{i,j}\) from the interaction-covering family with \(S_a\subseteq C(a)\)}.
Every $b\neq a$ such that $P_{b,C(a)}=P_a$ has $S_a\subsetneqq S_b$, so its coefficient $\alpha_{bb}$ has already been estimated. We therefore can compute ${\alpha}_{aa}$ through
\begin{equation}
\label{eq:back_substitution}
    {\alpha}_{aa}
    =
    {\gamma}_{P_a,C(a)}
    -
    \sum_{b\neq a: P_{b,C(a)}=P_a}{\alpha}_{bb}.
\end{equation}
For a maximal element $a$, the sum is empty and this reduces to ${\alpha}_{aa}={\gamma}_{P_a,C(a)}$. Continuing downward in the support-inclusion order estimates every $\alpha_{aa}$. If several clusters contain $S_a$, one may use any such cluster (averaging over them can reduce error in practice but does not make a big difference in the analysis). Since each term overlaps only $\Or(1)$ other terms under the bounded-degree locality assumption, each ${\alpha}_{aa}$ only involves $\Or(1)$ of the ${\alpha}_{bb}$. {Every substitution follows strict support inclusion, so dependency chains have length at most \(\mathsf k\). Bounded branching and bounded depth imply that \(\epsilon\) error on the restricted coefficients produces only \(\Or(\epsilon)\) error on each \(\alpha_{aa}\), and the total classical reconstruction takes \(\Or(N)\) time.} 

\subsubsection{Estimating \texorpdfstring{$\{\gamma_{P,C_{i,j}}\}$}{gamma sub {P,C {i,j}}} from experiments}
\label{sec:robust_est_gamma_P_C}

We now learn $\{\gamma_{P,C_{i,j}}\}$ from experiments, in parallel over the clusters of each partition $\Pi_i$. The protocol is robust under the SPAM noise model in Definition~\ref{defn:spam_noise_model}.

We first Pauli-twirl the evolution. For $R\in\bar{\mathcal{P}}_N$, let $\mathcal{R}[\rho]=R\rho R$. The twirled evolution is associated with the generator $\mathcal{L}^{\mathrm{diag}}$:
\begin{equation}
\label{eq:twirled_generator}
    \frac{1}{4^N}
    \sum_{R\in\bar{\mathcal{P}}_N}
    \mathcal{R}^{\dag}\circ\mathcal{L}\circ\mathcal{R}[\rho]=\sum_a \alpha_{aa}\left( P_{a} \rho P_{a} - \rho\right)=\mathcal{L}^{\mathrm{diag}}[\rho].
\end{equation}
Indeed, the Pauli average removes every Hamiltonian term and every dissipative term with distinct left and right Paulis.
Operationally, the twirled generator can be implemented by rapidly interleaving the evolution with uniformly random Pauli pulses.
Because the protocol only reads constant-size marginals, it is enough to control the Heisenberg evolution of local observables. 
This allows us to obtain a pulse rate that does not depend on the system size, as stated in the following lemma, which closely follows that in \cite[Theorem~16]{HuangTongFangSu2023learning}.

{
\begin{lemma}[System-size-independent local pulse-rate bound]
    \label{lem:local_pulse_rate}
    Given an $N$-qubit $(\mathsf k,\mathsf d)$-bounded degree local
    Lindbladian $\mathcal L=\sum_{\alpha\in V_{\mathrm{int}}}\mathcal
    L_\alpha$, let
    \begin{equation*}
        J=\max_{\alpha\in V_{\mathrm{int}}}
        \|\mathcal L_\alpha^\dagger\|_{\infty\to\infty},
    \end{equation*}
    {and suppose that \(J=\Or(1)\)}.  Define the ideal twirled channel
    $\mathcal U_t=e^{t\mathcal L^{\mathrm{diag}}}$ and the randomized-pulse
    channel
    \begin{equation}
        \label{eq:randomized_pulse_channel}
        \mathcal U'_{t,r}
        =
        \left(
            \frac{1}{4^N}
            \sum_{R\in\bar{\mathcal P}_N}
            \mathcal R^\dagger\circ e^{(t/r)\mathcal L}\circ\mathcal R
        \right)^r.
    \end{equation}
    There exist constants $K,\Delta_0>0$, depending only on
    $\mathsf k,\mathsf d$, and $J$, such that, whenever $t/r\leq\Delta_0$,
    every operator $O_C$ supported on a set $C$ with $|C|\leq\mathsf k$
    satisfies
    \begin{equation}
        \label{eq:local_pulse_heisenberg_bound}
        \left\|
            (\mathcal U'_{t,r})^\dagger[O_C]
            -
            \mathcal U_t^\dagger[O_C]
        \right\|_\infty
        \leq
        K\frac{t^2}{r}\|O_C\|_\infty.
    \end{equation}
\end{lemma}
}

As a result of this lemma, with the twirling described above, for every input state $\rho$,
    \begin{equation}
        \label{eq:local_pulse_marginal_bound}
        \left\|
            \Tr_{C^c}\!\left[
                \mathcal U'_{t,r}(\rho)-\mathcal U_t(\rho)
            \right]
        \right\|_1
        \leq
        K\frac{t^2}{r}.
    \end{equation}
    In particular, to make the local error (the left-hand side of \eqref{eq:local_pulse_marginal_bound}) at most $\epsilon$, we need
    $r\geq\max\{t/\Delta_0,Kt^2/\epsilon\}$, with no dependence on $N$.

{We also use a partial-twirl extension of Lemma~\ref{lem:local_pulse_rate}. Let \(\mathcal H\subseteq\bar{\mathcal P}_N\) be a product Pauli subgroup and let \(\overline{\mathcal L}=|\mathcal H|^{-1}\sum_{R\in\mathcal H}\mathcal R^\dagger\circ\mathcal L\circ\mathcal R\). Replace the full Pauli average in \eqref{eq:randomized_pulse_channel} by the average over \(\mathcal H\), and replace \(\mathcal U_t\) by \(e^{t\overline{\mathcal L}}\). If \(\overline{\mathcal L}^{\dagger}\) preserves the algebra of operators supported on \(C\), then \eqref{eq:local_pulse_heisenberg_bound} and \eqref{eq:local_pulse_marginal_bound} hold for that \(C\) with the same constants. The proof in Appendix~\ref{app:local_pulse_rate} applies with these replacements.}

{The proof, given in Appendix~\ref{app:local_pulse_rate}, adapts the
local Heisenberg-picture argument of Ref.~\cite{HuangTongFangSu2023learning}
to Lindbladian evolution.  
For the Lindbladian considered here, the coefficient bounds $|h_c|,|\alpha_{ab}|\leq1$ imply $J\leq2$, as required in Lemma~\ref{lem:local_pulse_rate}.
Since the experiment below uses a fixed evolution
time $t=\Or(1)$, $r=\Or(\epsilon^{-1})$ parallel Pauli-pulse rounds per shot
suffice to make the pulse-induced bias $\Or(\epsilon)$, independently of the
system size.}
{Each pulse round uses \(\Or(N)\) single-qubit gates, with
\(\Or(N)\) additional gates for preparation and measurement basis
changes and Pauli projections. Thus an experiment with \(r\) pulse
rounds uses \(\Or(N(r+1))\) gates.}

Fix a cluster $C$ and a Pauli operator $Q$ supported on $C$. Define
\begin{equation}
\label{eq:local_pauli_decay_rate}
    \ell_{Q,C}
    =
    \sum_{P:\supp(P)\subseteq C}
    \left(
        (-1)^{\braket{Q,P}}-1
    \right)\gamma_{P,C}
    =
    -2
    \sum_{\substack{P:\supp(P)\subseteq C\\
                    \braket{Q,P}=1}}
    \gamma_{P,C}.
\end{equation}
{Each qubit belongs to at most \(\mathsf d+1\) term supports, so at most \(\mathsf k(\mathsf d+1)\) diagonal terms can overlap \(\supp(Q)\). Since \(0\leq\alpha_{aa}\leq1\), the original coefficients give the bound

\begin{equation}
\label{eq:decay_rate_bound}
    0\leq-\ell_{Q,C}
=2\sum_{a:\braket{Q,P_a}=1}\alpha_{aa}
\leq2\mathsf k(\mathsf d+1)=\Or(1).
\end{equation}

}
By Lemma~\ref{lem:restricted_pauli_lindbladian},
\begin{equation}
    e^{t(\mathcal{L}_{\mathrm{diag}})^\dag}[Q]
    =
    e^{t\ell_{Q,C}}Q.
\end{equation}
Thus, it suffices to estimate $\ell_{Q,C}$ for every Pauli $Q$ supported on $C$.

We next describe a SPAM-robust experiment for estimating $\ell_{Q,C}$. Let $S=\supp(Q)$ and define
$
    Z_S=\prod_{a\in S}\sigma_a^z.
$
Choose a product Clifford $V_Q$ (i.e., a product of single-qubit Clifford gates) supported on $S$ such that
\begin{equation}
    V_Q Z_S V_Q^\dag=Q.
\end{equation}
According to the SPAM noise model in Definition~\ref{defn:spam_noise_model}, 
experimental noisy state preparation results in the state $\rho_0$. We then apply $V_Q$, resulting in a state
\begin{equation}
    \rho_Q=V_Q\rho_0V_Q^\dag.
\end{equation}
At the end of the experiment we will measure $Q$ by first applying $V_Q^\dag$ and measure in the computational basis. We use the measurement result to estimate the expectation value of $V_Q Z_S V_Q^\dag$, but because of the measurement error channel, the actual observable whose expectation value we estimate is
\begin{equation}
    M_Q
    =
    V_Q\mathcal{E}^\dag[Z_S]V_Q^\dag.
\end{equation}
The SPAM noise channels result in decay of the signal, which is captured by the following parameters:
\begin{equation}
\label{eq:spam_prefactors}
    s_Q
    =
    \Tr[Q\rho_Q]
    =
    \Tr[Z_S\rho_0],\quad 
    m_Q
    =
    \overline{\Tr}[QM_Q]
    =
    \overline{\Tr}[Z_S\mathcal{E}^\dag[Z_S]].
\end{equation}
Since $|S|\leq\mathsf{k}$, Definition~\ref{defn:spam_noise_model} implies
\begin{equation}
\label{eq:spam_prefactor_bounds}
    s_Q\geq r_p,
    \quad
    m_Q\geq r_m.
\end{equation}

The state $\rho_Q$ and observable $M_Q$ may contain Pauli components other than $Q$. We remove these components by an additional projection implemented by randomized twirling. Sampling two independent, uniformly random Paulis
$
    R_0,R_1\in\bar{\mathcal{P}}_N,
$
we apply $R_0$ immediately before the evolution and $R_1$ immediately after the evolution. We then multiply the measured parity (the $\pm 1$ measurement result we obtain from computational basis measurements) on $S$ by
$
(-1)^{\braket{Q,R_0}+\braket{Q,R_1}}.
$
Character orthogonality for the Pauli group gives
\begin{equation}
\begin{aligned}
    \mathbb{E}_{R_0}
    \left[
        (-1)^{\braket{Q,R_0}}
        R_0\rho_QR_0
    \right]
    =
    \frac{s_Q}{2^N}Q,\quad
    \mathbb{E}_{R_1}
    \left[
        (-1)^{\braket{Q,R_1}}
        R_1M_QR_1
    \right]
    =
    m_QQ.
\end{aligned}
\end{equation}
Consequently, denoting by $f_{Q,C}(t)$ the expectation of the signed random variable, i.e., the measured parity multiplied by $(-1)^{\braket{Q,R_0}+\braket{Q,R_1}}$, we have
\begin{equation}
\label{eq:spam_robust_decay_signal}
    f_{Q,C}(t)
    =
    s_Qm_Qe^{t\ell_{Q,C}}.
\end{equation}

For any fixed $\tau>0$, taking the ratio of the signals at times $0$ and $\tau$ gives
\begin{equation}
\label{eq:spam_robust_decay_estimator}
    \ell_{Q,C}
    =
    \frac{1}{\tau}
    \log\left(
        \frac{f_{Q,C}(\tau)}
             {f_{Q,C}(0)}
    \right).
\end{equation}
The unknown preparation and measurement factors cancel exactly. Moreover, \eqref{eq:spam_prefactor_bounds} implies
\begin{equation}
\label{eq:zero_time_signal_bound}
    f_{Q,C}(0)
    =
    s_Qm_Q
    \geq
    r_pr_m>0.
\end{equation}
Thus, the signal remains bounded away from zero by a constant independent of the system size.

Let $\widehat{f}_{Q,C}(0)$ and $\widehat{f}_{Q,C}(\tau)$ denote the corresponding empirical averages, and define
\begin{equation}
\label{eq:def_log_ratio}
    \widehat{\ell}_{Q,C}
    =
    \frac{1}{\tau}
    \log\left(
        \frac{\widehat{f}_{Q,C}(\tau)}
             {\widehat{f}_{Q,C}(0)}
    \right).
\end{equation}
{We use this formula when both empirical averages are positive, and set \(\widehat{\ell}_{Q,C}=0\) otherwise.}
{Choose \(\tau=[2\mathsf k(\mathsf d+1)]^{-1}\), which satisfies} 
$
    -\tau\ell_{Q,C}\leq 1
$
for every required pair $Q,C$. Then
\begin{equation}
    f_{Q,C}(\tau)
    \geq
    e^{-1}r_p r_m.
\end{equation}
{Set \(c=e^{-1}r_pr_m\) and \(\eta=\min\{c/2,\tau c\epsilon/4\}\). On the event that all empirical signals at times \(0\) and \(\tau\) differ from their expectations by at most \(\eta\), both the empirical and true signals are at least \(c/2\). Since \(\log\) is \(2/c\)-Lipschitz on \([c/2,\infty)\), this event implies 
\begin{equation}
\begin{aligned}
\left|\widehat{\ell}_{Q,C}-\ell_{Q,C}\right|
&\leq \frac{1}{\tau}\Bigl(
    \left|\log\widehat{f}_{Q,C}(\tau)-\log f_{Q,C}(\tau)\right|
    +\left|\log\widehat{f}_{Q,C}(0)-\log f_{Q,C}(0)\right|
\Bigr)\\
&\leq \frac{2}{\tau c}\Bigl(
    \left|\widehat{f}_{Q,C}(\tau)-f_{Q,C}(\tau)\right|
    +\left|\widehat{f}_{Q,C}(0)-f_{Q,C}(0)\right|
\Bigr)\\
&\leq \frac{4\eta}{\tau c}
\leq \epsilon,
\end{aligned}
\end{equation}
simultaneously for all required pairs \(Q,C\). 
The complementary event, including any nonpositive empirical signals, is counted as a failure. Applying Hoeffding's inequality to each empirical signal and taking a union bound over both sampling times \(t\in\{0,\tau\}\) and all \(\Or(N)\) required pairs \(Q,C\) shows that
}
\begin{equation}
\label{eq:decay_rate_sample_complexity}
    \Or\left(
        \epsilon^{-2}
        \log\left(N/\delta\right)
    \right)
\end{equation}
repetitions with the fixed evolution time $\tau=\Or(1)$ suffice to ensure
$
    \left|
        \widehat{\ell}_{Q,C}
        -
        \ell_{Q,C}
    \right|
    \leq\epsilon,
$
simultaneously for every required local Pauli $Q$ and every cluster $C$, with probability at least $1-\delta$.

It remains to recover $\gamma_{P,C}$ from the decay rates. Character orthogonality applied to \eqref{eq:local_pauli_decay_rate} gives, for every nonidentity Pauli $P$ supported on $C$,
\begin{equation}
\label{eq:gamma_from_decay_rates}
    \widehat{\gamma}_{P,C}
    =
    4^{-|C|}
    \sum_{Q:\supp(Q)\subseteq C}
    (-1)^{\braket{P,Q}}
    \widehat{\ell}_{Q,C},
    \quad
    \widehat{\ell}_{I,C}=0.
\end{equation}
We set $\widehat{\ell}_{I,C}=0$ because the identity operator is the fixed point of the diagonal Lindbladian $\mathcal{L}_{\mathrm{diag}}$.
Since the right-hand side of \eqref{eq:gamma_from_decay_rates} is a normalized sum, an additive error of at most $\epsilon$ in each decay rate  $\ell_{Q,C}$ produces an additive error of at most $\epsilon$ in each coefficient $\gamma_{P,C}$.

Finally, the protocol can be performed in parallel over the clusters of a partition. Because the clusters $C_{i,j}$ are disjoint, a single experimental shot can prepare and measure one chosen local Pauli on every cluster simultaneously. There are at most
$
    4^{\mathsf{k}}-1=\Or(1)
$
choices of the nonidentity local Pauli $Q$ in each cluster. Cycling through these choices, which involve an $\Or(1)$ overhead, therefore estimates all coefficients
$
\{\gamma_{P,C_{i,j}}\}_{P,j}
$
using the number of repetitions in \eqref{eq:decay_rate_sample_complexity}. 

\begin{lemma}
    \label{lem:estimate_gamma}
    Under the SPAM noise model in Definition~\ref{defn:spam_noise_model},
    for each partition $\Pi_i$ indexed by $i\in[\mathsf{r}]$,
    we can estimate all $\{\gamma_{P,C_{i,j}}\}_{P,j}$ for each nonidentity Pauli $P$ supported on cluster $C_{i,j}$ to within additive error $\epsilon$ with probability at least $1-\delta$ using $\Or(\epsilon^{-2}\log(N/\delta))$ experiments and total evolution time, for any $\epsilon,\delta>0$.
\end{lemma}

\subsubsection{The cost of learning the diagonal part}
\label{sec:cost_diagonal_part}

We are now ready to estimate the cost of estimating all coefficients $\alpha_{aa}$ associated with the Lindbladian in \eqref{eq:lindbladian_defn}. We first apply the method in the previous section to learn $\{\gamma_{P,C_{i,j}}\}_{P,j}$ in a SPAM-robust way for each $i\in[\mathsf{r}]$. Next, we use $\{\gamma_{P,C_{i,j}}\}_{P,j}$ to reconstruct $\alpha_{aa}$ using the back-substitution procedure described in \eqref{eq:back_substitution}. Estimation error in $\{\gamma_{P,C_{i,j}}\}_{P,j}$ can only be amplified $\Or(1)$ times in this procedure. Therefore, we have the following theorem.
\begin{theorem}
    \label{thm:estimate_diagonal}
    Under the SPAM noise model in Definition~\ref{defn:spam_noise_model}, for any $\epsilon,\delta>0$,
    we can estimate all $\{\alpha_{aa}\}_{a}$ to within additive error $\epsilon$ with probability at least $1-\delta$ using $\Or(\epsilon^{-2}\log(N/\delta))$ experiments and total evolution time.
    {The protocol uses
    \(\Or(N\epsilon^{-3}\log(N/\delta))\) single-qubit gates in total.}
\end{theorem}

\subsection{Gauge-invariant off-diagonal parameters}

The {universally} gauge-invariant off-diagonal {$\chi$-matrix components} can be divided into two categories: Type I and Type II, which correspond to the learnable real and imaginary components, respectively, as defined in Definitions~\ref{defn:typei_coeff} and~\ref{defn:typeii_coeff}.

\subsubsection{Type I}

The technique for learning Type I off-diagonals follows directly from that used to learn the diagonal coefficients. To demonstrate this, we will rotate the original Lindbladian such that the off-diagonal coefficients appear along the diagonal. 
We will write $P\cdot Q$ to denote the superoperator $\rho\mapsto P\rho Q$.
Consider the restriction of the Lindbladian $\mathcal L$ to the span of $P_a \cdot P_a, P_a \cdot P_b, P_b \cdot P_a, P_b \cdot P_b$, for $a, b \in \mathcal T_1$:
$$\mathcal L_{ab} = \alpha_{aa}P_a \rho P_a + \alpha_{ab}P_a \rho P_b + \alpha_{ab}^*P_b \rho P_a + \alpha_{bb}P_b \rho P_b.$$
By definition, we have that for $P_a = \prod_{i = 1} ^n P_{a, i}$ and $P_b = \prod_{i = 1} ^n P_{b, i}$, and there exists a unique $j$ such that $P_{a, j} \ne P_{b, j}$. 
We define the single-qubit Pauli acting at the $j$th index of $P_a$ as $Q$ and that acting at the $j$th index of $P_b$ as $R$. Since we only aim to learn the real part of $\alpha_{ab}$, which is the same as that of $\alpha_{ba}$, we only need to consider $(Q, R) \in \{(X, Y), (Y, Z), (Z, X)\}$. For each of these cases, we can define a single-qubit unitary transformation $U_{Q, R}$ that maps
\begin{equation}
    \label{eq:transformation_requirement}
    Q \mapsto \frac{1}{\sqrt 2}(Q - R), \quad R \mapsto \frac{1}{\sqrt 2}(Q + R), \quad Q \star R \mapsto Q \star R
\end{equation}
under conjugation {\(U_{Q,R}^\dagger\cdot U_{Q,R}\)}. Defining $U_{Q, R, j}$ as the action of $U_{Q, R}$ on the $j$th qubit,  we can define the unitary channel $\mathcal U_{Q, R, j}^\dagger = U_{Q, R, j} \cdot U_{Q, R, j}^\dagger$. We can decompose $U_{Q, R, j} = C_{Q, R, j}^\dagger TC_{Q, R, j}$, where 
\begin{equation}
\label{eq:type_i_dressing_choice}
    T=\begin{pmatrix}
    1 & 0 \\
    0 & e^{i\pi/4}
\end{pmatrix},\quad
C_{Q, R, j} = \begin{cases}
    I_j & (Q, R) = (X, Y) \\
    H_j & (Q, R) = (Y, Z) \\
    H_jS_jH_j & (Q, R) = (Z, X)
\end{cases},
\end{equation}
One can readily verify that this construction satisfies \eqref{eq:transformation_requirement}.
Since conjugating a Lindbladian by a unitary channel produces another Lindbladian, we arrive at the following lemma:

\begin{lemma}\label{lem:typei_coeff}
    For a Lindbladian $\mathcal L[\rho] = -i\sum_{c \in \mathcal I_h}h_c[P_c, \rho] + \sum_{(a, b) \in \mathcal I_{d}} \alpha_{ab}\mathcal D_{ab}[\rho]$, define the rotated Lindbladian $\mathcal L'[\rho] = \mathcal U_{Q, R, j} \circ \mathcal L \circ \mathcal U_{Q, R, j}^\dagger$. We then have that
    {
$$
\mathcal L'[\rho]=-i\sum_c h_c'[P_c,\rho]+\sum_{a,b}\alpha_{ab}'\mathcal D_{ab}[\rho],
$$
where \(c\) ranges over all nonidentity Pauli strings, and \((a,b)\) ranges over all ordered pairs of nonidentity Pauli strings, whose joint support is contained in an allowed patch
of the support pattern of \(\mathcal L\).
Here we allow coefficients that were absent from the original lists.}
    {For any} $(k, \ell) \in \mathcal T_1$ such that $P_{k, j} = Q, P_{\ell, j} = R$, we have that
    \begin{align*}
        \alpha_{kk}' =
        \frac{\alpha_{kk} + \alpha_{\ell\ell}}{2} + \operatorname{Re}[\alpha_{k \ell}], \quad 
        \alpha_{\ell \ell}' =
        \frac{\alpha_{kk} + \alpha_{\ell\ell}}{2} - \operatorname{Re}[\alpha_{k \ell}]
    \end{align*}
    and for all other $m$ for which $P_{m, j} \notin \{Q, R\}$,
    $\alpha_{mm}' = \alpha_{mm}$.
\end{lemma}

\begin{proof}
    Conjugation by single-qubit unitaries maps each Pauli string into a linear combination of Pauli strings with the same support, so $\mathcal L'$ remains within the support pattern of $\mathcal L$.
    If $(k, \ell) \in \mathcal T_1$ and $P_{k, j} = Q, P_{\ell, j} = R$, then under $\mathcal U_{Q, R, j}$,
    $$P_k \mapsto \frac{P_k - P_\ell}{\sqrt 2}, \quad P_\ell \mapsto \frac{P_k + P_\ell}{\sqrt 2},$$
    meaning that the restriction $\mathcal L_{k \ell}$ of the Lindbladian $\mathcal L$ to the span of $P_k \cdot P_k, P_k \cdot P_\ell, P_\ell \cdot P_k, P_\ell \cdot P_\ell$ maps to
    {
$$
\begin{aligned}
\mathcal L'_{k\ell}[\rho]
={}&\frac{\alpha_{kk}+\alpha_{\ell\ell}+\alpha_{k\ell}+\alpha_{\ell k}}{2}P_k\rho P_k\\
&+\frac{\alpha_{kk}+\alpha_{\ell\ell}-\alpha_{k\ell}-\alpha_{\ell k}}{2}P_\ell\rho P_\ell
+(\text{off-diagonals}),
\end{aligned}
$$
}
    giving us the values for $\alpha'_{kk}$ and {\(\alpha'_{\ell\ell}\)}. Furthermore, for any Pauli $P_m$ where $P_{m, j} \in \{I, Q \star R\}$, we have that $P_m \mapsto P_m$, meaning that we trivially have that $\alpha'_{mm} = \alpha_{mm}$.
\end{proof}

We will be applying the above transformation to multiple qubits for parallel learning.
We can construct a set of defect qubits $\mathcal Q_{i, s}$ such that $\forall (j, s), |\mathcal Q_{i, s} \cap C_{i, j}| \le 1$ and $\bigsqcup_{s} \mathcal Q_{i, s}$ is the set of all $N$ qubits. From Lemma~\ref{lem:typei_coeff}, we know that $\alpha_{kk}' + \alpha_{\ell \ell}' = \alpha_{kk} + \alpha_{\ell \ell}$. {
For a fixed nonidentity Pauli \(P\) on \(C_{i,j}\), the restricted coefficient (obtained through \eqref{eq:from_alpha_aa_to_gamma_PC})
\begin{equation}
    \label{eq:restricted_rotated_coefs}
\gamma'_{P,C_{i,j},\mathcal Q_{i,s}}
=
\sum_{R\in\bar{\mathcal P}_{[N]\setminus C_{i,j}}}
\alpha'_{P\otimes R,P\otimes R}
\end{equation}
is unchanged whether or not rotations are applied outside \(C_{i,j}\), provided the rotation inside \(C_{i,j}\) is held fixed. Individual diagonal coefficients in this sum may change.
Hence any defect on a qubit \(p\notin C_{i,j}\) has no effect on the restricted coefficient} $\gamma'_{P, C_{i, j}, \mathcal Q_{i, s}}$, defined as the coefficient $\gamma_{P, C_{i, j}}$ for the Lindbladian $\mathcal L' = \mathcal U_{Q, R} \circ \mathcal L \circ \mathcal U_{Q, R}^\dag$, where $\mathcal U_{Q, R} = \prod_{j \in \mathcal Q_{i, s}} \mathcal U_{Q_j, R_j, j}$ and $Q, R \in \bar{\mathcal P}_N$. 
Consequently, the transformed cluster coefficient \(\gamma'_{P,C_{i,j},\mathcal Q_{i,s}}\) depends on \(\mathcal Q_{i,s}\) only through the unique qubit \(q\in\mathcal Q_{i,s}\cap C_{i,j}\), whose local rotation transforms the restricted Lindbladian as described in Lemma~\ref{lem:typei_coeff}.

We now show how to recover the original Type I components from the rotated restricted coefficients \eqref{eq:restricted_rotated_coefs}.
For Pauli strings \(A,B\) on \(C=C_{i,j}\) that agree except at the distinguished qubit \(q\), where \(A_q=Q\) and \(B_q=R\), Lemma~\ref{lem:typei_coeff} and \eqref{eq:restricted_rotated_coefs}, together with invariance under exterior rotations, give
$$
\frac{\gamma'_{A,C,\mathcal Q_{i,s}}-\gamma'_{B,C,\mathcal Q_{i,s}}}{2}
=
\sum_{\substack{(k,\ell)\in\mathcal T_1\\
P_{k,C}=A,\;P_{\ell,C}=B}}
\Re\alpha_{k\ell}.
$$
Support-ordered back-substitution then recovers the individual components as in
\eqref{eq:back_substitution}. 

We next verify that the same decay-rate and pulse-rate bounds (Eq.~\eqref{eq:decay_rate_bound} and Lemma~\ref{lem:local_pulse_rate}) apply to the rotated Lindbladian.
Local unitary conjugation preserves the supports and induced norms of the original local summands. Although their Pauli expansions may acquire additional terms and coefficients larger than one, their norms remain at most \(2\). Twirling is an average of norm-preserving conjugations, so the same local decay-rate bound \(2\mathsf k(\mathsf d+1)\) and pulse-rate argument apply. The number of Pauli terms per support increases by at most a factor depending on \(\mathsf k\), preserving constant reconstruction overhead.

This means that to learn the Type I coefficients in parallel, we simply learn $\{\gamma'_{P, C_{i, j}, \mathcal Q_{i, s}}\}_{P, j}$ for all $s$ and for every $i \in [r]$. Since each cluster is of size at most $\mathsf{k}$, this means that there are at most $\mathsf{k} = \Or(1)$ choices of $s$. Hence, learning the Type I off-diagonals results in a constant overhead over Theorem~\ref{thm:estimate_diagonal}.
\begin{theorem}
    \label{thm:estimate_typei}
    Under the SPAM noise model in Definition~\ref{defn:spam_noise_model}, for any $\epsilon,\delta>0$,
    we can estimate all {\(\{\Re\alpha_{ab}\}_{(a,b)\in\mathcal T_1}\)} to within additive error $\epsilon$ with probability at least $1-\delta$ using $\Or(\epsilon^{-2}\log(N/\delta))$ experiments and total evolution time. 
    {The protocol uses
    \(\Or(N\epsilon^{-3}\log(N/\delta))\) single-qubit gates in total.}
\end{theorem}

\subsubsection{Type II}

While rotating the Lindbladian is sufficient to learn the Type I coefficients, the same does not hold true for Type II components. This follows from the fact that there is no such unitary transformation that maps $I$ to anything but $I$.
Consequently, we instead utilize partial twirling to only eliminate some of the off-diagonals. In this section we sometimes represent single-qubit Pauli matrices as $X,Y,Z$ for notation simplicity.

Fix $(a,b)\in\mathcal T_2$, and let $q$ be the unique qubit on which $P_a$ and $P_b$ differ. Without loss of generality, we assume
\begin{equation}
    P_a=I_q\otimes R,
    \quad
    P_b=Z_q\otimes R.
\end{equation}
We will show how to learn the imaginary part of the corresponding $\chi$-matrix entry $\chi_{I_q\otimes R,Z_q\otimes R}=\alpha_{ab}$.
{We order each pair with the identity on the left at \(q\); reversing the pair only changes the sign because \(\Im\alpha_{ba}=-\Im\alpha_{ab}\). The choice of the \(Z\) axis is without loss of generality, since the other two axes are related by trusted single-qubit Clifford rotations.}

Let $C$ be a cluster containing $S_{(a,b)}$, and let $U,V\in\{I,X,Y,Z\}$. For $R\in\bar{\mathcal P}_{C\setminus\{q\}}$, define the restricted coefficient
\begin{equation}
    \label{eq:typeii_restricted_chi}
    \chi^{(q,C)}_{UV;R}
    =
    \sum_{Q\in \bar{\mathcal{P}}_{[N]\setminus C}}
    \chi_{U_q\otimes R\otimes Q,V_q\otimes R\otimes Q},
\end{equation}
where $R$ is supported on $C\setminus\{q\}$ and $Q$ is supported on $[N]\setminus C$.
As in the diagonal case, these restricted coefficients combine all terms that have the same restriction to $C$. The original coefficients will be reconstructed by back-substitution after the local experiment.

{We partially Pauli-twirl the evolution to retain the Type II entries while reducing the relevant dynamics to a two-dimensional block.}
Let
\begin{equation}
    \mathcal G_q
    =
    \{I_q,Z_q\}\otimes\bar{\mathcal P}_{[N]\setminus\{q\}},
\end{equation}
and, for $K\in\mathcal G_q$, let $\mathcal K[\rho]=K\rho K$. 
The ideal partially twirled generator is
\begin{equation}
    \label{eq:typeii_partial_twirl}
    \mathcal L^{(q)}
    =
    \frac{1}{2\cdot4^{N-1}}
    \sum_{K\in\mathcal G_q}
    \mathcal K^\dagger\circ\mathcal L\circ\mathcal K.
\end{equation}
The physical protocol uses finite-rate randomized Pauli pulses to approximate this ideal evolution. 
We bound the resulting local signal bias below, after establishing the parallel partial-twirl construction.

The full Pauli twirl on every qubit other than $q$ retains only entries whose left and right Pauli labels agree away from $q$. The $\{I,Z\}$ twirl on $q$ retains the two sectors $\{I,Z\}$ and $\{X,Y\}$.
{Thus a surviving entry has \((U,V)\in\{I,Z\}^2\cup\{X,Y\}^2\) on \(q\). The Type II entries in the \(\{I,Z\}\) sector determine the antisymmetric part of the coupling between \(X\) and \(Y\); the \(\{X,Y\}\) sector also survives and will be canceled through classical postprocessing in the estimator below.}

For $R'\in\bar{\mathcal P}_{C\setminus\{q\}}$, define the Pauli Fourier coefficients
\begin{equation}
    \label{eq:typeii_fourier_coefficients}
    \widehat\chi^{R'}_{UV}
    =
    \sum_{R\in\bar{\mathcal P}_{C\setminus\{q\}}}
    (-1)^{\braket{R,R'}}\chi^{(q,C)}_{UV;R}.
\end{equation}
We view $R'$ as acting trivially outside $C\setminus\{q\}$. 
{To compute the action of the Heisenberg-picture generator \((\mathcal L^{(q)})^\dagger\) on \(\sigma_q\otimes R'\), where \(\sigma\in\{X,Y\}\), note that a term with left and right labels  \(U_q\otimes R\otimes Q\) and \(V_q\otimes R\otimes Q\) contributes
$$
\chi_{U_q\otimes R\otimes Q,V_q\otimes R\otimes Q}
(V\sigma U)_q\otimes(RR'R)\otimes(QI Q)
$$
to \((\mathcal L^{(q)})^\dagger[\sigma_q\otimes R']\), for \(\sigma\in\{X,Y\}\). 
In the Heisenberg picture, using \(\chi_{U,V}^{*}=\chi_{V,U}\),  we have \(\mathcal L^\dagger[O]=\sum_{U,V}\chi_{U,V}VOU\).
Since \(RR'R=(-1)^{\braket{R,R'}}R'\) and \(QI Q=I\), summing the entries that survive the partial Pauli twirl (i.e., those in $\{I,Z\}^2\cup\{X,Y\}^2$) gives
$$
(\mathcal L^{(q)})^\dagger[\sigma_q\otimes R']
=\left(
\sum_{(U,V)\in\{I,Z\}^2\cup\{X,Y\}^2}
\widehat\chi^{R'}_{UV}V\sigma U
\right)_q\otimes R'.
$$
This uses the full \(\chi\) representation, so the Hamiltonian and anticommutator contributions are already included. For each surviving pair, \(V\sigma U\) is proportional to \(X\) or \(Y\), proving that their span is invariant. In particular, the \((I,Z)\) and \((Z,I)\) entries contribute \(-2\Im\widehat\chi^{R'}_{IZ}Y\) for \(\sigma=X\) and \(+2\Im\widehat\chi^{R'}_{IZ}X\) for \(\sigma=Y\); the \((X,Y)\) and \((Y,X)\) entries contribute \(2\Re\widehat\chi^{R'}_{XY}\) to both off-diagonal couplings. Combining these with the diagonal entries yields}
\begin{equation}
    \label{eq:typeii_effective_action}
    \begin{aligned}
        (\mathcal L^{(q)})^\dagger[X_q\otimes R']
        &=
        (d+c)X_q\otimes R'
        +{(a-b)}Y_q\otimes R',\\
        (\mathcal L^{(q)})^\dagger[Y_q\otimes R']
        &=
        {(a+b)}X_q\otimes R'
        +(d-c)Y_q\otimes R',
    \end{aligned}
\end{equation}
where
\begin{equation}
    \label{eq:typeii_block_parameters}
    a=2\Re\widehat\chi^{R'}_{XY},
    \quad
    b=2\Im\widehat\chi^{R'}_{IZ},
    \quad
    c=\widehat\chi^{R'}_{XX}-\widehat\chi^{R'}_{YY},
    \quad
    d=\widehat\chi^{R'}_{II}-\widehat\chi^{R'}_{ZZ}.
\end{equation}
In particular, if
\begin{equation}
    \mathbf P_{R'}
    =
    \begin{pmatrix}
        X_q\otimes R'\\
        Y_q\otimes R'
    \end{pmatrix},
\end{equation}
then
\begin{equation}
    \label{eq:typeii_block_matrix}
    (\mathcal L^{(q)})^\dagger[\mathbf P_{R'}]
    =A_{R'}\mathbf P_{R'},
    \quad
    A_{R'}={aX-ibY+cZ+dI}
    =
    \begin{pmatrix}
        d+c&{a-b}\\
        {a+b}&d-c
    \end{pmatrix}.
\end{equation}
This gives a two-dimensional effective dynamics even though the original Lindbladian acts on $N$ qubits.

The time dependence of this block can be written explicitly. Defining
\begin{equation}
    \Delta=a^2-b^2+c^2,
    \quad
    \lambda=\sqrt{\Delta},
\end{equation}
we have
\begin{equation}
    \label{eq:typeii_matrix_exponential}
    e^{tA_{R'}}
    =
    e^{dt}
    \left(
        \cosh(\lambda t)I
        +
        \frac{\sinh(\lambda t)}{\lambda}
        {(aX-ibY+cZ)}
    \right).
\end{equation}
The expression is understood by continuity when $\Delta=0$. When $\Delta<0$, writing $\lambda=i\omega$ turns the hyperbolic functions into $\cos(\omega t)$ and $\sin(\omega t)$.

We next describe a SPAM-robust experiment for estimating \(b=2\Im\widehat{\chi}^{R'}_{IZ}\). The formula below eliminates \(a\), while \(c\) and \(d\) do not enter the off-diagonal derivatives at \(t=0\).
The procedure is similar to the one in Section~\ref{sec:robust_est_gamma_P_C}. 
Fix $R'\in\bar{\mathcal P}_{C\setminus\{q\}}$ and let
\begin{equation}
    Q_{\mu,R'}=\sigma_q^\mu\otimes R',
    \quad
    \mu\in\{x,y\}.
\end{equation}
Let $S=\supp(Q_{\mu,R'})=\{q\}\cup\supp(R')$, which is independent of $\mu$, and define $Z_S=\prod_{j\in S}\sigma_j^z$. For each $\mu\in\{x,y\}$, choose a product Clifford $V_{\mu,R'}$ supported on $S$ such that
\begin{equation}
    V_{\mu,R'}Z_SV_{\mu,R'}^\dagger=Q_{\mu,R'}.
\end{equation}
Starting from the noisy computational-basis state $\rho_0$, the state used to prepare the $\mu$ component is
\begin{equation}
    \rho_{\mu,R'}
    =
    V_{\mu,R'}\rho_0V_{\mu,R'}^\dagger.
\end{equation}
Similarly, measuring $Q_{\mu,R'}$ is implemented by applying $V_{\mu,R'}^\dagger$ and measuring in the computational basis on $S$. Accounting for the measurement-error channel $\mathcal E$, the corresponding effective observable is
\begin{equation}
    M_{\mu,R'}
    =
    V_{\mu,R'}\mathcal E^\dagger[Z_S]V_{\mu,R'}^\dagger.
\end{equation}
We define the preparation and measurement visibilities by
\begin{equation}
    \begin{aligned}
        s_{R'}
        &=
        \Tr[Q_{\mu,R'}\rho_{\mu,R'}]
        =
        \Tr[Z_S\rho_0],\\
        m_{R'}
        &=
        \overline{\Tr}[Q_{\mu,R'}M_{\mu,R'}]
        =
        \overline{\Tr}[Z_S\mathcal E^\dagger[Z_S]].
    \end{aligned}
\end{equation}
The right-hand sides are independent of $\mu$. Thus, the $X$ and $Y$ experiments have the same preparation visibility $s_{R'}$ and the same measurement visibility $m_{R'}$.

For \(\mu,\nu\in\{x,y\}\), we first describe the ideal experiment, whose expected signal is
\begin{equation}
\label{eq:F_mu_nu_defn}
    F_{\mu\nu}^{R'}(t)
    =
    m_{R'}
    \left(e^{tA_{R'}}\right)_{\mu\nu}
    s_{R'}.
\end{equation}
We prepare $\rho_{\nu,R'}$, apply a uniformly random Pauli
$R_0\in\bar{\mathcal P}_N$, evolve for time $t$ under the ideal generator $\mathcal L^{(q)}$, and then apply an independent uniformly random Pauli $R_1\in\bar{\mathcal P}_N$. 
We finally measure $Q_{\mu,R'}$ using the basis change described above. 
We denote by $Y_{\mu\nu}^{R'}(t)\in\{-1,1\}$ the measurement outcome, and multiply it in classical post-processing by
$
    (-1)^{
        \braket{Q_{\nu,R'},R_0}
        +
        \braket{Q_{\mu,R'},R_1}
    }.
$
Character orthogonality gives
\begin{equation}
\label{eq:typeii_pauli_projections}
    \begin{aligned}
        \mathbb E_{R_0}\left[
            (-1)^{\braket{Q_{\nu,R'},R_0}}
            R_0\rho_{\nu,R'}R_0
        \right]
        &=
        \frac{s_{R'}}{2^N}Q_{\nu,R'},\\
        \mathbb E_{R_1}\left[
            (-1)^{\braket{Q_{\mu,R'},R_1}}
            R_1M_{\mu,R'}R_1
        \right]
        &=
        m_{R'}Q_{\mu,R'}.
    \end{aligned}
\end{equation}
Thus, the two random Pauli operations remove every unwanted Pauli component of the prepared state and measured observable.
{The expected signed outcome is therefore
$$
\begin{aligned}
&\mathbb{E}[(-1)^{\braket{Q_{\nu,R'},R_0}+\braket{Q_{\mu,R'},R_1}}Y_{\mu\nu}^{R'}(t)]
=m_{R'}s_{R'}\overline{\Tr}\!\left[
\left(e^{t(\mathcal L^{(q)})^\dagger}[Q_{\mu,R'}]\right)Q_{\nu,R'}
\right]
=m_{R'}s_{R'}(e^{tA_{R'}})_{\mu\nu}
=F_{\mu\nu}^{R'}(t).
\end{aligned}
$$
The final equality uses Pauli orthogonality and Eq.~\eqref{eq:typeii_block_matrix}. Here $\mu$ labels the measured Pauli and $\nu$ labels the prepared Pauli.}
Since $S\subseteq C$ and $|C|\leq\mathsf{k}$, the SPAM noise model in Definition~\ref{defn:spam_noise_model} implies
\begin{equation}
    s_{R'}\geq r_p,
    \quad
    m_{R'}\geq r_m.
\end{equation}

Collecting the four experiments corresponding to $(\mu,\nu)\in\{x,y\}^2$ into a matrix gives
\begin{equation}
    \label{eq:typeii_spam_signal_matrix}
    F^{R'}(t)
    =
    \begin{pmatrix}
        F_{xx}^{R'}(t)&F_{xy}^{R'}(t)\\
        F_{yx}^{R'}(t)&F_{yy}^{R'}(t)
    \end{pmatrix}
    =
    m_{R'}s_{R'}e^{tA_{R'}}.
\end{equation}
The common SPAM prefactor cancels in the combination needed below. Differentiating at the origin gives
\begin{equation}
    \begin{aligned}
        \dot F_{xy}^{R'}(0)&=m_{R'}s_{R'}{(a-b)},\quad
        \dot F_{yx}^{R'}(0)&=m_{R'}s_{R'}{(a+b)},\quad
        F_{xx}^{R'}(0)&=F_{yy}^{R'}(0)=m_{R'}s_{R'}.
    \end{aligned}
\end{equation}
It follows that
\begin{equation}
    \label{eq:typeii_direct_b}
    b
    =
    \frac{
        {\dot F_{yx}^{R'}(0)-\dot F_{xy}^{R'}(0)}
    }{
        2F_{xx}^{R'}(0)
    }.
\end{equation}

We estimate the derivatives in \eqref{eq:typeii_direct_b} by interpolating each signal at Chebyshev nodes and differentiating the interpolating polynomial. The following lemma records the required accuracy and sample complexity.

\begin{lemma}[Chebyshev derivative estimation]
    \label{lem:chebyshev_derivative_estimation}
    Let $f:[0,T]\to[-1,1]$, and suppose that
    $
        g(x)=f\left(\frac{T}{2}(x+1)\right)
    $
    extends analytically to the Bernstein ellipse $E_\rho$ for some $\rho>1$, with $|g(z)|\leq B$ throughout $E_\rho$. Suppose that, at any chosen $t\in[0,T]$, one can sample a random variable in $[-1,1]$ with expectation $f(t)$. For every integer $n\geq1$, there is an estimator $\widehat D_n$ of $f'(0)$ using the $n+1$ mapped Chebyshev--Lobatto nodes such that
    \begin{equation}
        \left|\mathbb E[\widehat D_n]-f'(0)\right|
        \leq
        \frac{C_\rho B}{T}n^2\rho^{-n}.
    \end{equation}
    Moreover, using a total of $M$ samples allocated among these nodes, with probability at least $1-\delta$,
    \begin{equation}
        \left|\widehat D_n-\mathbb E[\widehat D_n]\right|
        \leq
        \frac{Cn^2}{T}
        \sqrt{\frac{\log(2/\delta)}{M}},
    \end{equation}
    where $C$ and $C_\rho$ are independent of $f,n,M,$ and $N$. Consequently, when $T,B,$ and $\rho$ are constants, choosing $n=\Theta(\log(1/\epsilon))$ estimates $f'(0)$ to additive error $\epsilon$ using
    \begin{equation}
        \Or\left(
            \epsilon^{-2}
            \log^4(2/\epsilon)
            \log(1/\delta)
        \right)
    \end{equation}
    samples. 
    All mapped Chebyshev--Lobatto nodes lie in \([0,T]\), so \(M\) samples use cumulative evolution time at most \(TM\). 
    Since \(T=\Or(1)\), total evolution time has the same asymptotic order as the sample count.
\end{lemma}

\begin{proof}
    Let $p_n$ be the degree-$n$ polynomial interpolating $g$ at the Chebyshev--Lobatto nodes $x_j=\cos(j\pi/n)$. Analyticity on $E_\rho$ implies that the Chebyshev coefficients of $g$ decay as $\Or(B\rho^{-k})$ \cite[Theorem~8.1]{trefethen2019approximation}. Since
    $
        |T_k'(-1)|=k^2,
    $
    the standard Chebyshev interpolation and aliasing bounds \cite[Theorem~4.2 and Equation~(4.9)]{trefethen2019approximation} give
    \begin{equation}
        \left|\frac{2}{T}p_n'(-1)-f'(0)\right|
        \leq
        \frac{C_\rho B}{T}
        \sum_{k>n}k^2\rho^{-k}
        \leq
        \frac{C_\rho B}{T}n^2\rho^{-n}.
    \end{equation}
    The endpoint, i.e., the point $-1$, derivative of the interpolant is a linear combination of the sampled values,
    \begin{equation}
        \frac{2}{T}p_n'(-1)
        =
        \sum_{j=0}^{n}w_j f\left(\frac{T}{2}(x_j+1)\right),
    \end{equation}
    where by explicit calculation using \cite[Equation (21.2)]{trefethen2019approximation}, the endpoint row of the Chebyshev differentiation matrix satisfies
    \begin{equation}
        \sum_{j=0}^{n}|w_j|\leq \frac{Cn^2}{T}.
    \end{equation}
    {Let \(W=\sum_j|w_j|\). For each of the \(M\) independent samples, draw a node index \(j\) with probability \(|w_j|/W\), sample its bounded outcome \(Y_j\), and record \(W\sgn(w_j)Y_j\). The sample mean \(\widehat D_n\) has expectation \(\sum_jw_j f(T(x_j+1)/2)=2p_n'(-1)/T\). 
    {Each recorded value lies in \([-W,W]\), so Hoeffding's inequality gives}}
    \begin{equation}
        \left|\widehat D_n-\mathbb E[\widehat D_n]\right|
        \leq
        \left(\sum_{j=0}^{n}|w_j|\right)
        \sqrt{\frac{2\log(2/\delta)}{M}}
    \end{equation}
    with probability at least $1-\delta$. The claimed complexity follows by taking $n=\Theta(\log(1/\epsilon))$ so that the interpolation bias is at most $\epsilon/2$, and then choosing $M$ so that the statistical error is also at most $\epsilon/2$.
\end{proof}

For the signals in \eqref{eq:typeii_spam_signal_matrix}, bounded-degree locality and bounded local coefficients imply $\|A_{R'}\|=\Or(1)$. 
{Indeed, only the \(O(1)\) local terms intersecting \(\{q\}\cup\operatorname{supp}(R')\) act nontrivially on \(X_q\otimes R'\) or \(Y_q\otimes R'\), and each has bounded norm.} Equation~\eqref{eq:typeii_spam_signal_matrix} therefore extends analytically to complex time and is uniformly bounded on a fixed Bernstein ellipse for any fixed $T=\Or(1)$. This allows us to apply Lemma~\ref{lem:chebyshev_derivative_estimation} to $F_{xy}^{R'}$ and $F_{yx}^{R'}$. The zero-time diagonal signals can be estimated with the smaller cost $\Or(\epsilon^{-2}\log(1/\delta))$ as they do not require differentiation. Since
$
    F_{xx}^{R'}(0)
    =
    m_{R'}s_{R'}
    \geq
    r_pr_m,
$
the division in \eqref{eq:typeii_direct_b} amplifies estimation error by at most a constant independent of $N$.
{More explicitly, if each of the two derivatives and the zero-time signal is estimated to error at most \(\xi\leq r_pr_m/2\), the empirical denominator is positive and the resulting estimate satisfies \(|\widehat b-b|\leq 2(1+|b|)\xi/(r_pr_m)\). Since \(|b|=\Or(1)\), choosing \(\xi\) to be a sufficiently small constant multiple of \(\epsilon\) suffices.}

Once the signed values of $b$ have been recovered for every $R'\in\bar{\mathcal P}_{C\setminus\{q\}}$, by the definition \eqref{eq:typeii_fourier_coefficients}, the inverse Pauli Fourier transform gives
\begin{equation}
    \label{eq:typeii_inverse_fourier}
    \Im\chi^{(q,C)}_{IZ;R}
    =
    \frac{1}{2\cdot4^{|C|-1}}
    \sum_{R'\in\bar{\mathcal P}_{C\setminus\{q\}}}
    (-1)^{\braket{R,R'}}b_{R'}.
\end{equation}
{The absolute weights in this inverse transform sum to \(1/2\), so uniform error \(\xi\) in the \(b_{R'}\) gives error at most \(\xi/2\) in each restricted imaginary coefficient.}

For $R\neq I$, we write
\(P_a=I_q\otimes R\) and \(P_b=Z_q\otimes R\), where
\(R\in\bar{\mathcal P}_{C\setminus\{q\}}\). 
Then $(P_a,P_b)$ is a Type II pair supported in \(C\) corresponding to a dissipative term.
Eq.~\eqref{eq:typeii_restricted_chi} gives
\begin{equation}
\Im\chi^{(q,C)}_{IZ;R}
=
\sum_{\substack{(k,\ell)\in\mathcal T_2\\
P_{k,C}=I_q\otimes R,\;
P_{\ell,C}=Z_q\otimes R}}
\Im\alpha_{k\ell}.
\end{equation}
Every summand other than \(\Im\alpha_{ab}\) has a nonidentity
common Pauli factor outside \(C\), so its joint support strictly
contains that of \((a,b)\).
We therefore recover the individual $\Im\alpha_{ab}$ by the same support-inclusion back-substitution used in \eqref{eq:back_substitution}, processing larger supports before smaller supports.
Repeating over the interaction-covering clusters, the
distinguished qubits, and the three local Pauli axes allows us
to estimate all Type II imaginary dissipative coefficients $\Im \alpha_{k\ell}$.

For $R=I$, we have
\begin{equation}
\Im\chi^{(q,C)}_{IZ;I}
=h_{Z_q}+
\sum_{\substack{(k,\ell)\in\mathcal T_2\\
P_{k,C}=I_q\otimes I,\;
P_{\ell,C}=Z_q\otimes I}}
\Im\alpha_{k\ell}.
\end{equation}
{The \(Q=I\) summand in Eq.~\eqref{eq:typeii_restricted_chi} is \(\Im\chi_{I,Z_q}=h_{Z_q}\), without contribution from dissipative coefficients, because the combined contribution of the anticommutator terms \(-\frac12\sum_{a,b}\alpha_{ab}\{P_bP_a,\cdot\}\) affects only the real part of \(\chi_{I,Z_q}\), since \(\sum_{a,b}\alpha_{ab}P_bP_a\) is Hermitian.
Every \(Q\neq I\) summand has two nonidentity Pauli labels and is therefore a Type II dissipative coefficient.}
Each dissipative coefficient in this sum has nontrivial support
outside \(C\) and has already been recovered by the \(R\neq I\)
procedure on a cluster containing its full support.
Repeating over all qubits and the three local Pauli axes
therefore recovers every single-qubit Hamiltonian coefficient.

{For error propagation, each back-substitution involves only \(\Or(1)\) coefficients whose supports overlap the target support, and every dependency follows strict support inclusion. The dependency depth is at most \(\mathsf k\), so the cumulative error amplification is bounded independently of \(N\). Recovering \(h_{Z_q}\) requires one further subtraction of \(\Or(1)\) already estimated dissipative coefficients. Thus all recovered Type II components have error at most \(\epsilon\) when the local estimates have error a sufficiently small constant multiple of \(\epsilon\).}

{Finally, the protocol can be performed in parallel with constant overhead. For each interaction-covering partition $\Pi_i$, bounded-degree locality of the Lindbladian and bounded cluster size imply that each cluster shares an interaction term with only $\Or(1)$ other clusters. Hence, the clusters can be grouped into $\Or(1)$ subfamilies such that no Lindbladian term intersects two clusters in the same subfamily. In each round, choose one distinguished qubit $q$ per cluster in the active subfamily, perform the $\{I,Z\}$ twirl (or its counterpart in another Pauli basis) on those qubits, and fully Pauli-twirl all remaining qubits. This separation preserves the local invariant subspaces used above. Cycling over the subfamilies, the at most $\mathsf{k}$ choices of $q$, the three local Pauli axes, and the $4^{|C|-1}=O(1)$ choices of $R'$ incurs only constant overhead.}

{The Chebyshev differentiation procedure in Lemma~\ref{lem:chebyshev_derivative_estimation} is compatible with paralleization, because Chebyshev nodes and their sampling probabilities depend only on \(n\) and \(T\), so each shot can use the same sampled time for all active clusters and yield one signed parity for each cluster. For each cluster, independence across experimental shots gives concentration of its signal estimate. A union bound then gives a simultaneous guarantee for all clusters, even though their outcomes within the same shot may be correlated.}

{We now verify that the single-distinguished-qubit analysis and
the finite-pulse error bound of Lemma~\ref{lem:local_pulse_rate}
remain valid for the parallel partial-twirl protocol.
For an active cluster \(C\), no term meeting \(C\) contains a partially twirled qubit outside \(C\). 
Consequently, for every term meeting \(C\), the full Pauli twirl
on its support outside \(C\) retains only equal left and right
exterior Pauli factors \(Q\). When acting on
\(O_C\otimes I_{C^c}\), these cancel in the sandwich terms
because \(Q I_{C^c}Q=I_{C^c}\), and in the anticommutator terms
because \(Q^2=I_{C^c}\). Terms disjoint from \(C\) annihilate
this observable.
Thus the algebra on \(C\) is invariant, and its restriction agrees with the single-distinguished-qubit calculation above. The partial-twirl extension of Lemma~\ref{lem:local_pulse_rate} therefore bounds each signal's finite-pulse bias by \(Kt^2/r\). 
{Using the signed Pauli-averaging identities in Eq.~\eqref{eq:typeii_pauli_projections}, the measured operator is \(m_{R'}Q_{\mu,R'}\), while the signed average of the randomized input states is \(s_{R'}Q_{\nu,R'}/2^N\), which has trace norm \(s_{R'}\leq1\). Since \(m_{R'}\leq1\), the local Heisenberg bound in the partial-twirl extension of Lemma~\ref{lem:local_pulse_rate} therefore bounds each signal's finite-pulse bias by \(Kt^2/r\).}
By the proof of Lemma~\ref{lem:chebyshev_derivative_estimation},
the derivative estimator has total absolute weight
\(\sum_{j=0}^{n}|w_j|=\Or(n^2/T)\), where
\(n=\Theta(\log(1/\epsilon))\) is the interpolation degree and
\([0,T]\) is the sampling interval, with \(T>0\) fixed
independently of \(N\) and \(\epsilon\).
Thus, a uniform signal bias of \(\Or(\epsilon T/n^2)\)
contributes only \(\Or(\epsilon)\) to the derivative error.
For fixed \(T\), \(r=\Or(n^2/\epsilon)\) pulse rounds per experiment suffice. This changes neither the experiment count nor the total evolution-time scaling.}
{Each such experiment uses \(\Or(Nn^2/\epsilon)\)
single-qubit gates.}

We therefore have the following theorem:
\begin{theorem}
    \label{thm:estimate_typeii}
    Under the SPAM noise model in Definition~\ref{defn:spam_noise_model}, for any $\epsilon,\delta>0$, we can estimate all $\{\Im\alpha_{ab}\}_{(a,b)\in\mathcal T_2}$ {and all $\{h_c:|\supp(P_c)|=1\}$} to within additive error $\epsilon$ with probability at least $1-\delta$ using
    \begin{equation}
        \Or\left(
            \epsilon^{-2}
            \log^4(2/\epsilon)
            \log(N/\delta)
        \right)
    \end{equation}
    experiments and total evolution time of the same order. 
    {The protocol uses
    \(\Or(N\epsilon^{-3}\log^6(2/\epsilon)\log(N/\delta))\)
    single-qubit gates in total.}
\end{theorem}
In the above the $\log(N/\delta)$ factor arises due to a union bound over the $\Or(N)$ required local signals. {The implicit constants depend only on \(\mathsf k,\mathsf d,r_p,r_m\).}

{\section{Boundary-induced gauge invariance of Hamiltonian coefficients}
\label{sec:hamiltonian_learnability}}

{Section~\ref{sec:gauge_transformations_freedom} classifies components as universally gauge invariant or generically gauge dependent among physical Lindbladians, using the ambient classification of Theorem~\ref{thm:gauge_invariant_coefficients} and the physical genericity result of Corollary~\ref{cor:generic_physical_gauge_dependence}. Here we study a different notion: \emph{boundary-induced gauge invariance} at a fixed physical Lindbladian. The requirement that both generators produce CPTP semigroups imposes positivity of their Kossakowski matrices. 
As an example, for boundary instances with zero diagonal entries in the Kossakowski matrix, the positivity restricts the directions in which gauge transformation can modify the Lindbladian while preserving physicality.
Under the sufficient conditions below, these constraints force additional Hamiltonian coefficients to remain unchanged between physical gauge-equivalent generators. Such coefficients remain generically gauge dependent; their invariance at these instances is not universal but instance-specific. Learning these additional boundary-induced invariants in general remains an open problem.}

More specifically, we now use positivity of the
Kossakowski matrix to give simple sufficient conditions under which a Pauli
Hamiltonian coefficient cannot change between two gauge-equivalent
Lindbladians.  Throughout this section, let
\begin{equation}
    \mathcal L'
    =
    \mathcal G\circ\mathcal L\circ\mathcal G^{-1},
    \qquad
    \mathcal G[P]
    =
    g_{\supp(P)}P,
    \label{eq:general_gauge_hamiltonian_section}
\end{equation}
where \(\mathcal G\) is an invertible Hermiticity-preserving commuting gauge
transformation of the form in Eq.~\eqref{eq:general_commuting_gauge}.  We
suppose that both \(\mathcal L\) and \(\mathcal L'\) are Lindbladians.  In this
section we index the Kossakowski matrix directly by nonidentity Pauli
operators, and write \(\alpha_{P,Q}\) and \(\alpha'_{P,Q}\) for the
Kossakowski coefficients of \(\mathcal L\) and \(\mathcal L'\), respectively.
We similarly write \(h_c\) and \(h'_c\) for the coefficients of \(P_c\) in
their Hamiltonians.

Define the set of Pauli operators with nonzero diagonal Kossakowski
coefficient by
\begin{equation}
    \mathsf D(\mathcal L)
    =
    \left\{
        P\in\bar{\mathcal P}_N\setminus\{I\}:
        \alpha_{P,P}>0
    \right\}.
    \label{eq:diagonal_pauli_set}
\end{equation}
Positivity of the Kossakowski matrix gives
\begin{equation}
    |\alpha_{P,Q}|^2
    \leq
    \alpha_{P,P}\alpha_{Q,Q}.
    \label{eq:kossakowski_two_by_two_bound}
\end{equation}
Consequently, if \(P\notin\mathsf D(\mathcal L)\), every element of the row
and column indexed by \(P\) vanishes.

We will use the following observation about the diagonal coefficients.
For Pauli operators \(U,V\), let \(U\mathbin{\cdot}V\) denote the
superoperator \(\rho\mapsto U\rho V\).

\begin{lemma}
\label{lem:diagonal_kossakowski_invariant}
For every \(P\in\bar{\mathcal P}_N\setminus\{I\}\), we have
$
    \alpha'_{P,P}
    =
    \alpha_{P,P}.
$
Hence
\(\mathsf D(\mathcal L')=\mathsf D(\mathcal L)\).
\end{lemma}

\begin{proof}
The term \(U\mathbin{\cdot}V\) maps an input Pauli operator \(B\) to a
multiple of the Pauli operator with projective label
\(U\star V\star B\).  Thus every nonzero Pauli-transfer matrix element
\(\lvert A)(B\rvert\) of this term satisfies
\begin{equation}
    A\star B
    =
    U\star V.
\end{equation}
Similarity transformation by \(\mathcal G\) multiplies this matrix element by
\(g_{\supp(A)}/g_{\supp(B)}\), and therefore does not change
\(A\star B\).  It follows that gauge similarity transformation can mix
\(U\mathbin{\cdot}V\) only with terms for which the product of the left and
right Pauli labels is also \(U\star V\).

When \(U\star V=I\), one has \(U=V\).  Moreover,
\(U\mathbin{\cdot}U\) maps every Pauli operator to itself up to a sign, so it
commutes with \(\mathcal G\) and is fixed by gauge similarity.  The
coefficient of \(P\mathbin{\cdot}P\) in the left--right Pauli expansion of a
Lindbladian is \(\alpha_{P,P}\).  The preceding observations therefore give $\alpha'_{P,P}=\alpha_{P,P}$.
\end{proof}

We first give a condition that applies to a Hamiltonian term of any weight.

\begin{theorem}[General sufficient condition]
\label{thm:hamiltonian_product_condition}
Suppose \(g_S>0\) for every \(S\subseteq[N]\).  If, for all nonidentity
Pauli operators \(P,Q\),
\begin{equation}
    P\star Q=P_c
    \quad\Longrightarrow\quad
    P\notin\mathsf D(\mathcal L)
    \ \text{or}\
    Q\notin\mathsf D(\mathcal L),
    \label{eq:hamiltonian_product_condition}
\end{equation}
then
$
    h'_c=h_c.
$
\end{theorem}

\begin{proof}
Equation~\eqref{eq:kossakowski_two_by_two_bound} and the hypothesis imply
\begin{equation}
    \alpha_{P,Q}=0
    \qquad
    \text{whenever }
    P\star Q=P_c.
    \label{eq:no_product_dissipator}
\end{equation}
Lemma~\ref{lem:diagonal_kossakowski_invariant} shows that
\(\mathcal L'\) has the same set \(\mathsf D\), so the corresponding
coefficients \(\alpha'_{P,Q}\) also vanish.

For each projective Pauli operator \(R\), let
\[
    \mathsf V_R
    =
    \operatorname{span}
    \left\{
        U\mathbin{\cdot}V:
        U\star V=R
    \right\}.
\]
Indeed,
\[
    \mathcal D_{P,Q}
    =
    P\mathbin{\cdot}Q
    -\frac{1}{2}(QP)\mathbin{\cdot}I
    -\frac{1}{2}I\mathbin{\cdot}(QP)
\]
belongs entirely to \(\mathsf V_{P\star Q}\), while
\(-i[P_c,\mathbin{\cdot}]\) belongs to \(\mathsf V_{P_c}\).
The fact that $\alpha_{P,Q}=\alpha'_{P,Q}=0$ whenever $P\star Q=P_c$ and $P,Q\in \mathsf{D}(\mathcal{L})$ therefore
show that the \(\mathsf V_{P_c}\) components of \(\mathcal L\) and
\(\mathcal L'\) are, respectively,
\(-ih_c[P_c,\mathbin{\cdot}]\) and
\(-ih'_c[P_c,\mathbin{\cdot}]\).
The specific commutator form is due to the assumption that both $\mathcal{L}$ and $\mathcal{L}'$ are valid Lindbladians and therefore preserve the trace and Hermiticity.
The proof of Lemma~\ref{lem:diagonal_kossakowski_invariant} shows that gauge
similarity preserves every subspace \(\mathsf V_R\). Projecting
\(\mathcal L'=\mathcal G\circ\mathcal L\circ\mathcal G^{-1}\) onto
\(\mathsf V_{P_c}\) therefore gives
\begin{equation}
    \mathcal G\circ
    \bigl(-ih_c[P_c,\mathbin{\cdot}]\bigr)
    \circ\mathcal G^{-1}
    =
    -ih'_c[P_c,\mathbin{\cdot}].
    \label{eq:isolated_hamiltonian_similarity}
\end{equation}

If \(h_c=0\), this equation immediately gives \(h'_c=0\).  Otherwise, choose
a Pauli operator \(R\) that anticommutes with \(P_c\), and let
\(T=P_c\star R\).  The commutator in
Eq.~\eqref{eq:isolated_hamiltonian_similarity} contains both transitions
\(R\mapsto T\) and \(T\mapsto R\).  Gauge similarity multiplies their
coefficients by
\begin{equation}
    r
    =
    \frac{g_{\supp(T)}}{g_{\supp(R)}}
    \qquad\text{and}\qquad
    r^{-1},
\end{equation}
respectively.  Both factors must equal \(h'_c/h_c\), so \(r=r^{-1}\).
Since every \(g_S\) is positive, \(r=1\), and hence \(h'_c=h_c\).
\end{proof}

For odd-weight Hamiltonian terms, we will propose a less restrictive sufficient condition.
Let \(S_c=\supp(P_c)\).  For \(R\subseteq S_c\), define \(P_{c,R}\) to agree
with \(P_c\) on \(R\) and to be the identity outside \(R\).  Thus
\begin{equation}
    P_{c,R}\star P_{c,S_c\setminus R}
    =
    P_c.
    \label{eq:complementary_restriction_product}
\end{equation}

\begin{theorem}[Odd-weight sufficient condition]
\label{thm:hamiltonian_odd_weight_condition}
We assume $g_S\neq 0$ for all $S\subseteq [N]$.
Suppose \(|S_c|\) is odd.  If
\begin{equation}
    \left\{
        P_{c,R},P_{c,S_c\setminus R}
    \right\}
    \nsubseteq
    \mathsf D(\mathcal L)
    \label{eq:hamiltonian_odd_weight_condition}
\end{equation}
for every nonempty proper subset \(R\subsetneq S_c\), then
$
h'_c=h_c.
$
\end{theorem}

\begin{proof}
We consider
\begin{equation}
    \mathcal W_c
    =
    \left(
        \bigotimes_{j\in S_c}
        [P_{c,j},\mathbin{\cdot}]
    \right)
    \otimes
    \mathcal I_{S_c^c},
    \label{eq:odd_weight_superoperator}
\end{equation}
where $I_{S_c^c}$ is the identity superoperator acting on $S_c^c$.
Its left--right Pauli expansion is
\begin{equation}
    \mathcal W_c
    =
    \sum_{R\subseteq S_c}
    (-1)^{|S_c\setminus R|}
    P_{c,R}\mathbin{\cdot}P_{c,S_c\setminus R}.
    \label{eq:odd_weight_superoperator_expansion}
\end{equation}
At a qubit \(j\in S_c\), the corresponding local factor annihilates a Pauli
operator that commutes with \(P_{c,j}\), and maps either anticommuting Pauli
operator to a nonzero scalar multiple of another nonidentity Pauli operator.
It follows that \(\mathcal W_c\) maps every Pauli operator either to zero or
to a scalar multiple of a Pauli operator with exactly the same support.  Hence
both \(\mathcal W_c\) and \(\mathcal W_c^\dagger\) commute with \(\mathcal G\).

Equip the space of superoperators with the Hilbert--Schmidt inner product
\begin{equation}
    \langle\mathcal A,\mathcal B\rangle_{\mathrm{sop}}
    =
    \Tr(\mathcal A^\dagger\mathcal B).
\end{equation}
Cyclicity of the superoperator trace gives
\begin{equation}
\begin{aligned}
    \langle\mathcal W_c,\mathcal L'\rangle_{\mathrm{sop}}
    =
    \Tr\left(
        \mathcal W_c^\dagger
        \mathcal G\mathcal L\mathcal G^{-1}
    \right)=
    \Tr\left(
        \mathcal G^{-1}
        \mathcal W_c^\dagger
        \mathcal G\mathcal L
    \right)=
    \langle\mathcal W_c,\mathcal L\rangle_{\mathrm{sop}}.
\end{aligned}
\label{eq:odd_weight_pairing_invariant}
\end{equation}

For the unnormalized Pauli basis,
\begin{equation}
    \left\langle
        U\mathbin{\cdot}V,
        U'\mathbin{\cdot}V'
    \right\rangle_{\mathrm{sop}}
    =
    4^N\delta_{U,U'}\delta_{V,V'}.
    \label{eq:left_right_pauli_orthogonality}
\end{equation}
Since \(|S_c|\) is odd, the coefficients of
\(P_c\mathbin{\cdot}I\) and \(I\mathbin{\cdot}P_c\) in
Eq.~\eqref{eq:odd_weight_superoperator_expansion} are \(1\) and \(-1\),
respectively.  Therefore
\begin{equation}
    \left\langle
        \mathcal W_c,
        -i[P_c,\mathbin{\cdot}]
    \right\rangle_{\mathrm{sop}}
    =
    -2i\,4^N.
    \label{eq:odd_weight_hamiltonian_overlap}
\end{equation}

We will then show that for a dissipative term $P\cdot Q- \frac{1}{2}\{QP,\cdot\}$, its inner product with $\mathcal{W}_c$ must be zero due to \eqref{eq:hamiltonian_odd_weight_condition}.
First, the two anticommutator contributions $-\frac{1}{2}QP\cdot I$ and $-\frac{1}{2}I\cdot QP$ in a dissipative term have equal
coefficients for the terms with left--right labels
\(P_c\mathbin{\cdot}I\) and \(I\mathbin{\cdot}P_c\), and hence cancel in
their inner product with \(\mathcal W_c\). 
Second, for term \(P\mathbin{\cdot}Q\) in the dissipator to have nonzero inner product with
\(\mathcal W_c\), we need
\begin{equation}
    (P,Q)
    =
    \left(
        P_{c,R},P_{c,S_c\setminus R}
    \right)
\end{equation}
for a nonempty proper subset \(R\subsetneq S_c\).  Because of the hypothesis \eqref{eq:hamiltonian_odd_weight_condition},
Eq.~\eqref{eq:kossakowski_two_by_two_bound}, and
Lemma~\ref{lem:diagonal_kossakowski_invariant} force all these Kossakowski
coefficients to vanish for both \(\mathcal L\) and \(\mathcal L'\).
It follows that the inner product between any dissipative term and $\mathcal{W}_c$ is necessarily zero. Consequently, the commutator with $P_c$ is the only term that contributes to the inner product with $\mathcal{W}_c$ in both $\mathcal{L}$ and $\mathcal{L}'$. We therefore have
\begin{equation}
    \langle\mathcal W_c,\mathcal L\rangle_{\mathrm{sop}}
    =
    \langle\mathcal W_c,-ih_c[P_c,\cdot]\rangle_{\mathrm{sop}}
    =
    -2i\,4^N h_c,
    \quad
    \langle\mathcal W_c,\mathcal L'\rangle_{\mathrm{sop}}
    =
    \langle\mathcal W_c,-ih'_c[P_c,\cdot]\rangle_{\mathrm{sop}}
    =
    -2i\,4^N h'_c.
\end{equation}
Combining these equations with
Eq.~\eqref{eq:odd_weight_pairing_invariant} proves \(h'_c=h_c\).
\end{proof}

We remark that $h_c$ for single-qubit $P_c$ is {universally gauge invariant} because it is a Type-II learnable coefficient, for which we have a learning protocol as described in Theorem~\ref{thm:estimate_typeii}.

\section{Conclusion}

In this work, we developed a framework connecting the gauge freedom caused by unknown SPAM noise to the limits and possibilities of learning local Lindbladians. For any prescribed pattern of bounded-degree local interactions, we identified all universally gauge-invariant components and gave efficient protocols to learn them using only trusted single-qubit control. For every remaining component, we showed they are generically gauge dependent among physical Lindbladians and could therefore retain an irreducible uncertainty even with unlimited data. Furthermore, we showed that this ambiguity could remain of constant size as the system size increased, despite complete-positivity constraints. Our protocols learn all universal invariants with nearly optimal dependence on the target accuracy and only logarithmic dependence on system size in the number of experiments and total evolution time, without ancillas or entangling control. These guarantees required only nonvanishing local SPAM visibility. We also identified sufficient conditions under which positivity enforced additional, instance-specific Hamiltonian invariants at the boundary of the physical set.

A central open problem is to turn boundary-induced gauge invariance into efficient learning protocols and to quantify how the resulting guarantees deteriorate when small dissipative perturbations move the generator away from the boundary. Beyond individual coefficients, characterizing and learning gauge-invariant combinations could reveal dynamical information that remains accessible even when their constituent parameters are not. It would also be useful to determine which additional calibrated operations suffice to remove specific gauge ambiguities, and how uncertainty in the single-qubit controls changes both identifiability and estimation accuracy. Extending the framework in these directions could provide a systematic way to match experimentally available controls to the information needed for noise characterization and the study of dissipative many-body dynamics.

Our sufficient conditions for boundary-induced Hamiltonian gauge invariance suggest an intriguing connection to the Hamiltonian-not-in-Lindblad-span condition for Heisenberg-limited quantum metrology \cite{Zhou2017AchievingTH}. The role of dissipative operator products in our conditions echoes their role in determining whether quantum error correction can preserve a Hamiltonian signal while suppressing noise. Establishing a precise relationship between these criteria remains an open problem, since gauge invariance alone does not guarantee Heisenberg-limited estimation. Such a connection could guide the development of SPAM-robust protocols for additional Hamiltonian coefficients and clarify which control resources would enable precision beyond the standard quantum limit.

\section*{Statement on the Use of Generative AI}

The classification results and learning algorithms that form the main
contributions of this work were developed and proved entirely by the authors
and were initially recorded in a collection of informal notes. ChatGPT
was used to polish and organize these notes into an initial manuscript
draft. It was also used to draft the proofs of 
Theorem~\ref{thm:uniform_gauge_width},
Proposition~\ref{prop:local_simultaneous_gauge_dependence}, and
Theorem~\ref{thm:hamiltonian_odd_weight_condition}, and to assist in extending
\cite[Theorem 16, Appendix D]{HuangTongFangSu2023learning} to the Lindbladian setting to obtain
Lemma~\ref{lem:local_pulse_rate}. All AI-assisted arguments and text were
independently checked, revised, and approved by the authors, who take full
responsibility for the content of the manuscript.

\section*{Acknowledgments} This material is based upon work supported by the U.S. Department of Energy, Office of Science, Accelerated Research in Quantum Computing Centers, Quantum Utility through Advanced Computational Quantum Algorithms, grant no. DE-SC0025572 (Y.T.). 
STF was supported by the U.S. Department of Energy, Office of Science, National Quantum Information Science Research Centers, Co-design Center for Quantum Advantage (C2QA) under contract number DE-SC0012704, and by the U.S. National Science Foundation National Quantum Virtual Laboratory (NQVL), Erasure Qubits and Dynamic Circuits for Quantum Advantage (ERASE), under grant number OSI-2435244.

\appendix

\section{Proof of the local pulse-rate bound}
\label{app:local_pulse_rate}

{
\begin{proof}[Proof of Lemma~\ref{lem:local_pulse_rate}]
Let
\begin{equation*}
    \Phi_\Delta
    =
    \frac{1}{4^N}
    \sum_{R\in\bar{\mathcal P}_N}
    \mathcal R^\dagger\circ e^{\Delta\mathcal L}\circ\mathcal R,
    \qquad
    \mathcal V_\Delta=e^{\Delta\mathcal L^{\mathrm{diag}}},
\end{equation*}
so that $\mathcal U'_{t,r}=\Phi_{t/r}^r$ and
$\mathcal U_t=\mathcal V_{t/r}^r$.  We first establish a one-step bound
on operators supported on $C$.

If $X$ is supported on $C$ and a local term $\mathcal L_\alpha$ is
supported on a set disjoint from $C$, then
$\mathcal L_\alpha^\dagger[X]=0$.  Hence a term
\begin{equation*}
    \mathcal L_{\alpha_j}^\dagger\cdots
    \mathcal L_{\alpha_1}^\dagger[X]
\end{equation*}
can be nonzero only if, at each step, the next local term overlaps $C$ or
one of the preceding terms.  Every qubit belongs to at most
$\mathsf d+1$ term supports, and every term overlaps at most $\mathsf d$
other terms. Each of these non-zero terms therefore correspond to a word $(\alpha_1,\cdots,\alpha_j)$. Denoting the number of these words $(\alpha_1, \dots, \alpha_j)$ as $A_j(C)$ having non-zero contribution, we have
\begin{align*}
    A_1(C) &\leq \mathsf k (\mathsf d + 1) \\
    A_j(C) &\leq A_{j - 1}(C)[\mathsf{k} (\mathsf d + 1) + (j - 1)\mathsf d] \leq A_{j - 1}(C)[(\mathsf{k} + j - 1) (\mathsf d + 1)],
\end{align*}
employing the fact that $|C|\leq \mathsf{k}$. Expanding this recursive bound gives us
\begin{equation}
    \label{eq:local_word_count}
    A_j(C)
    \leq
    \prod_{m=1}^j (\mathsf k + m-1)(\mathsf d+1)
    = j!\binom{\mathsf k + j - 1}{\mathsf k - 1}(\mathsf d + 1)^j
    \leq
    2^{\mathsf k-1}j!\,[2(\mathsf d+1)]^j,
\end{equation}
Pauli conjugation
preserves both the support and the induced norm of every local term.  Thus,
with
\begin{equation*}
    a_{\mathsf k}=2^{\mathsf k-1},
    \quad
    g=2(\mathsf d+1)J,
\end{equation*}
we have, uniformly over $R\in\bar{\mathcal P}_N$,
\begin{equation}
    \label{eq:local_lindbladian_power_bound}
    \left\|
        \left[
            (\mathcal R^\dagger\circ\mathcal L\circ\mathcal R)^\dagger
        \right]^j[X]
    \right\|_\infty
    \leq
    a_{\mathsf k}j!g^j\|X\|_\infty.
\end{equation}
{The same bound holds for the generator averaged over any product Pauli subgroup: average each local summand separately. Each averaged summand remains trace-annihilating, has support contained in its original support, and has induced norm at most \(J\). Thus the same local-word count applies, in particular to \(\mathcal L^{\mathrm{diag}}\).}

The zeroth-order terms in $\Phi_\Delta^\dagger[X]$ and
$\mathcal V_\Delta^\dagger[X]$ coincide, and their first-order terms
coincide by Eq.~\eqref{eq:twirled_generator}.  Expanding both channels and
using Eq.~\eqref{eq:local_lindbladian_power_bound}, for $g\Delta\leq1/2$
we obtain
\begin{equation}
    \label{eq:local_one_step_pulse_bound}
    \left\|
        \Phi_\Delta^\dagger[X]-\mathcal V_\Delta^\dagger[X]
    \right\|_\infty
    \leq
    2a_{\mathsf k}
    \sum_{j=2}^\infty(g\Delta)^j\|X\|_\infty
    \leq
    4a_{\mathsf k}g^2\Delta^2\|X\|_\infty.
\end{equation}
This bound depends on $\mathsf k,\mathsf d$, and $J$, but not on $N$.

Both $\Phi_\Delta$ and $\mathcal V_\Delta$ are quantum channels, so their
adjoints are contractions in operator norm.  Moreover,
$\mathcal V_\Delta^\dagger$ preserves the algebra of operators supported
on $C$, since $\mathcal L^{\mathrm{diag}}$ is diagonal in the Pauli basis.
The telescoping identity therefore gives
\begin{equation*}
\begin{aligned}
    \left\|
        (\Phi_\Delta^\dagger)^r[O_C]
        -(\mathcal V_\Delta^\dagger)^r[O_C]
    \right\|_\infty 
    \leq
    \sum_{u=0}^{r-1}
    \left\|
        (\Phi_\Delta^\dagger-\mathcal V_\Delta^\dagger)
        (\mathcal V_\Delta^\dagger)^u[O_C]
    \right\|_\infty
    \leq
    4a_{\mathsf k}g^2r\Delta^2\|O_C\|_\infty.
\end{aligned}
\end{equation*}
Setting $\Delta=t/r$ proves
Eq.~\eqref{eq:local_pulse_heisenberg_bound} with
$K=4a_{\mathsf k}g^2$ and $\Delta_0=(2g)^{-1}$.

\end{proof}
}

For the partial-twirl extension, we only need to make minimal modifications to the above proof. We use
$$
\begin{aligned}
\Phi_\Delta&=\frac{1}{|\mathcal H|}\sum_{R\in\mathcal H}\mathcal R^\dagger\circ e^{\Delta\mathcal L}\circ\mathcal R,\quad 
\mathcal V_\Delta&=e^{\Delta\overline{\mathcal L}}.
\end{aligned}
$$
Their first-order terms agree, and the preceding averaged-local-term bound gives the identical one-step estimate. Both maps are CPTP, and the assumed invariance of the algebra on \(C\) ensures that \((\mathcal V_\Delta^\dagger)^u[O_C]\) remains supported on \(C\). The same telescoping and duality arguments therefore give both claimed bounds.

\bibliographystyle{ieeetr}
\bibliography{refs}

\end{document}